\documentclass[lettersize,journal]{IEEEtran}
\usepackage{amsmath,amsfonts}
\usepackage{array}
\usepackage[caption=false,font=normalsize,labelfont=sf,textfont=sf]{subfig}
\usepackage{textcomp}
\usepackage{stfloats}
\usepackage{booktabs}
\usepackage{float}
\usepackage{graphicx}
\usepackage{amsmath}
\usepackage{amsthm}
\usepackage{amssymb}
\usepackage{algorithm}
\usepackage{algpseudocode}
\usepackage{afterpage}
\usepackage{url}
\usepackage{bm}
\usepackage{verbatim}
\usepackage{graphicx}
\usepackage{multirow}
\usepackage{makecell}
\usepackage{array}
\usepackage{stfloats}
\usepackage{booktabs}
\newtheorem{theorem}{Theorem}
\usepackage{color} % 使用color包
\usepackage{balance}
\usepackage[numbers,sort&compress]{natbib}
\def\BibTeX{{\rm B\kern-.05em{\sc i\kern-.025em b}\kern-.08em
		T\kern-.1667em\lower.7ex\hbox{E}\kern-.125emX}}
\usepackage{balance}
\usepackage[colorlinks,
linkcolor=black,
anchorcolor=black,
citecolor=black,
urlcolor=blue]{hyperref}

\begin{document}
	
\title{\huge{Pre-Chirp-Domain Index Modulation for Full-Diversity \\Affine Frequency Division Multiplexing towards 6G}}
\author{ 
		{Guangyao Liu, 
		Tianqi~Mao,~\IEEEmembership{Member,~IEEE},
		Zhenyu~Xiao,~\IEEEmembership{Senior Member,~IEEE},
		Miaowen~Wen,~\IEEEmembership{Senior Member,~IEEE},
		Ruiqi Liu,~\IEEEmembership{Senior Member,~IEEE}, 
		Jingjing Zhao,~\IEEEmembership{Member,~IEEE}, 
		Ertugrul Basar,~\IEEEmembership{Fellow,~IEEE}, 
		Zhaocheng Wang,~\IEEEmembership{Fellow,~IEEE},
		Sheng Chen,~\IEEEmembership{Life~Fellow,~IEEE}}
		
		% <-this % stops a space
  \thanks{This work was supported in part by the National Natural Science Foundation of China (NSFC) under grant numbers U22A2007, 62171010, 62088101 and 62401054, U2233216 and 62471015, in part by the Open Project Program of State Key Laboratory of CNS/ATM under Grant numbers 2024B12, in part by Young Elite Scientists Sponsorship Program by CAST under Grant numbers 2022QNRC001 and in part by the Fundamental Research Funds for the Central Universities under Grant numbers 2024ZYGXZR076. Part of this work has been presented in IEEE IWCMC 2024 \cite{10592488}. \emph{(Corresponding authors: Tianqi Mao, Zhenyu Xiao.)}}
	\thanks{G. Liu, Z. Xiao and J. Zhao are with the School of Electronic and Information Engineering and the State Key Laboratory of CNS/ATM, Beihang University, Beijing 100191, China (e-mails: liugy@buaa.edu.cn, xiaozy@buaa.edu.cn, jingjingzhao@buaa.edu.cn).}
	\thanks{T. Mao is with Greater Bay Area Innovation Research Institute of BIT, Zhuhai 519000, China, and is also with Beijing Institute of Technology (Zhuhai), Zhuhai 519088, China (e-mail: maotq@bit.edu.cn).}
	\thanks{R. Liu is with the Wireless and Computing Research Institute, ZTE	Corporation, Beijing 100029, China, and also with the State Key Laboratory of Mobile Network and Mobile Multimedia Technology, Shenzhen 518055, China (e-mail: richie.leo@zte.com.cn).}
	\thanks{M. Wen is with the School of Electronic and Information Engineering, South China University of Technology, Guangzhou 510640, China (e-mail: eemwwen@scut.edu.cn).}
	\thanks{E. Basar is with the Department of Electrical Engineering, Tampere University, 33720 Tampere, Finland, on leave from the Department of Electrical and Electronics Engineering, Koc University, 34450 Sariyer, Istanbul, Turkey (email: ertugrul.basar@tuni.fi and ebasar@ku.edu.tr).}
	\thanks{Z. Wang is with the Department of Electronic Engineering, Tsinghua University, Beijing 100084, China (e-mail: zcwang@tsinghua.edu.cn).}
	\thanks{S. Chen is with the School of Electronics and Computer Science, University of Southampton, Southampton SO17 1BJ, U.K. (e-mail: sqc@ecs.soton.ac.uk).}
\vspace{-.5cm}
}
	
	%\markboth{Journal of \LaTeX\ Class Files,~Vol.~18, No.~9, September~2020}%
	%{How to Use the IEEEtran \LaTeX \ Templates}
\maketitle
	
\begin{abstract}
		As a superior multicarrier technique utilizing chirp signals for high-mobility communications, affine frequency division multiplexing (AFDM) is envisioned to be a promising candidate for sixth-generation (6G) wireless networks. AFDM is based on the discrete affine Fourier transform (DAFT) with two adjustable parameters of the chirp signals, termed the pre-chirp and post-chirp parameters, respectively. Whilst the post-chirp parameter complies with stringent constraints to combat the time-frequency doubly selective channel fading, we show that the pre-chirp counterpart can be flexibly manipulated for an additional degree of freedom. Therefore, this paper proposes a novel AFDM scheme with the pre-chirp index modulation (PIM) philosophy (AFDM-PIM), which can implicitly convey extra information bits through dynamic pre-chirp parameter assignment, thus enhancing both spectral and energy efficiency. Specifically, we first demonstrate that the subcarrier orthogonality is still maintained by applying distinct pre-chirp parameters to various subcarriers in the AFDM modulation process. Inspired by this property, we allow each AFDM subcarrier to carry a unique pre-chirp signal according to the incoming bits. By such an arrangement, extra bits can be embedded into the index patterns of pre-chirp parameter assignment without additional energy consumption. We derive asymptotically tight upper bounds on the average bit error probability (BEP) of the proposed schemes with the maximum-likelihood detection, and validate that the proposed AFDM-PIM can achieve full diversity under doubly dispersive channels. Based on the derived result, we further propose an optimal pre-chirp alphabet design to enhance the bit error rate (BER) performance via intelligent optimization algorithms. Simulation results demonstrate that the proposed AFDM-PIM outperforms the classical benchmarks.
\end{abstract}
	
	
	% \newpage
\section{Introduction}\label{S1}

	\IEEEPARstart{T}{he} beyond fifth-generation (B5G) and sixth-generation (6G) wireless networks are envisioned to deliver ultra-reliable, high data rate, and low-latency communications for high--speed mobile scenarios, including low-earth-orbit (LEO) satellite, high-mobility railway, unmanned aerial vehicles (UAV) and vehicle-to-vehicle (V2V) communications \cite{liu2023beginning,10054381,9509294,9689960,Zhaojingjing2025,Mao2025WCL,Wang2023CN}. These scenarios inevitably suffer from severe Doppler shifts, which can cause time-frequency doubly selective fading (i.e., doubly dispersive channel) by involving the multi-path effects\cite{Zhao2025TWC}. This makes the existing modulation formats, like the main-stream orthogonal frequency division multiplexing (OFDM) in 4G/5G standards, no longer suitable for next-generation networks \cite{wu2016survey}. Consequently, it is crucial to develop new waveforms for next-generation communication networks to adapt to the doubly selective channel.
	
	To date, several novel modulation techniques have been designed to combat the time-frequency doubly selective fading, such as orthogonal time-frequency space (OTFS) \cite{7925924,10159363,9508932} and orthogonal chirp-division multiplexing (OCDM) \cite{9346006,10214498,10497118}. OTFS modulates information in the delay-Doppler (DD) domain using the inverse symplectic finite Fourier transform (ISFFT), which enables the transmission symbols to be multiplexed across the entire time-frequency domain \cite{8686339,8727425,9794710}. OCDM utilizes a series of orthogonal chirp signals whose frequency varies with time to modulate information, which achieves better performance than the OFDM technique under doubly dispersive channels. 
	However, the two-dimensional representation of the DD channel in OTFS incurs significant pilot overhead, and the diversity gain that OCDM can obtain depends on specific channel profiles. 
	
	Against this background, affine frequency division multiplexing (AFDM) has been proposed based on the discrete affine Fourier transform (DAFT) \cite{bemani2021afdm}, which can also combat the time-frequency doubly selective fading, and more importantly has less complexity to implement than OTFS since it requires only one-dimension transformation. DAFT is defined as a generalized discrete form of the discrete Fourier transform (DFT) with a chirp-like basis specified by dual adjustable parameters, termed pre-chirp and post-chirp parameters, respectively. In AFDM, data symbols are multiplexed onto chirp-like subcarriers through DAFT and inverse DAFT (IDAFT), which can separate the doubly dispersive channel into a sparse, quasi-static channel with a comprehensive DD channel representation by appropriately setting the chirp parameters. Therefore, AFDM achieves similar performance to OTFS, and demonstrates superior performance over OFDM and OCDM under doubly selective channels \cite{bemani2023affine}, making AFDM a promising waveform for future mobile communication systems \cite{10769778}.
	
	There has been preliminary literature on AFDM \cite{10557524,10566604,9940346,10551402,10580928,10439996}. A low-complexity embedded pilot-aided diagonal reconstruction (EPA-DR) channel estimation scheme was proposed in \cite{10557524}, which calculates the AFDM effective channel matrix directly without estimating the three channel parameters, eliminating the inherently severe inter-Doppler interference. In \cite{10566604}, the authors investigated the AFDM-empowered sparse code multiple access (SCMA) systems to support massive connectivity in high-mobility environments. An AFDM-based integrated sensing and communications (ISAC) system was studied in \cite{9940346}, demonstrating that the AFDM-ISAC system can maintain excellent sensing performance even under significant Doppler shifts. In \cite{10551402}, the authors proposed a bistatic sensing-aided channel estimation (SACE) scheme, which illustrates the capacity of AFDM to attain optimal sensing resolution and communication performance under reasonable parameter selection. A DFT-based modulation and demodulation technique for AFDM was studied in \cite{10580928}, with the potential to facilitate the adoption of AFDM for future deployments. The authors in \cite{10439996} showed that utilizing a single pilot symbol achieves nearly the same sensing performance as employing the complete AFDM frame and proposed a low-complexity self-interference cancellation scheme for monostatic radar. 
	The existing literature mostly explored the channel estimation, multiple access and ISAC issues under the classical AFDM architecture, whilst studies regarding further optimization/enhancement of the AFDM waveform are still at their infancy.

	One promising research direction is to incorporate the index modulation (IM) philosophy for spectral and energy efficiency improvement \cite{mao2018novel,8004416,Mao2024TWC,Liu2024TVT}, which conveys energy-free bits through the activation patterns of transmit entities, e.g., subcarriers \cite{bacsar2013orthogonal}, time slots \cite{9398861}, pulse positions \cite{wen2019survey}, antennas \cite{6823072}, constellation modes \cite{7547943}, etc. IM-aided systems provide unique advantages, including flexible system architectures, high spectral efficiency (SE), and low hardware complexity, by leveraging additional energy-free information bits compared to conventional modulation systems.
	
	In \cite{10570960,24wenJSAC}, Tao \emph{et al.} presented an IM-assisted scheme, which conveys energy-free information bits through the activation patterns of the subsymbols in the DAFT-domain, verifying that index bits have stronger diversity protection than modulation bits. A multicarrier system using the activation patterns of AFDM chirp subcarriers as indices was developed in \cite{10342712}, which indicates the potential of IM-assisted AFDM technology in enhancing bit error rate (BER) and energy efficiency performance. However, existing research has concentrated on post-chirp parameters in AFDM, with little attention paid to the considerable flexibility and degrees of freedom (DoFs) that the pre-chirp parameter offers. Although the methods of index modulation based on the value of the pre-chirp parameter were discussed in \cite{rou2024afdm}, it focused more on reducing the complexity of the demodulation algorithm at the receiver, and no complete design method for the value of different pre-chirp parameters was proposed in this study.
	
	Distinctively, this paper proposes a novel AFDM scheme with the pre-chirp-domain index modulation (AFDM-PIM) to enhance both spectral and energy efficiencies. Furthermore, performance evaluation of the proposed AFDM-PIM structure, including pairwise error probability (PEP) analysis and diversity analysis, is performed, and the numerical selection of the pre-chirp parameters is analyzed and optimized. The main contributions of this work are highlighted as follows:
	
\begin{itemize}
	\item We prove that the subcarrier orthogonality is maintained by applying distinct pre-chirp parameters to different subcarriers during the AFDM modulation process. Based on this property, each AFDM subcarrier is constructed with a unique pre-chirp signal corresponding to the incoming bits. This configuration allows for the embedding of additional bits into the index patterns of pre-chirp parameter assignment without additional energy consumption.
	\item We derive the input-output relationship of the proposed AFDM-PIM scheme in the DAFT domain, and the asymptotically tight upper bounds on the average bit error probability (BEP) with the maximum likelihood (ML) detection based on the PEP analysis. Furthermore, we validate that the proposed AFDM-PIM scheme can achieve full diversity order under doubly dispersive channels. 
	\item We propose an optimal pre-chirp alphabet design to enhance the BER performance via particle swarm optimization (PSO) algorithm. It is verified via extensive simulations that the optimized pre-chirp parameter alphabet results in a much better BER performance than the heuristic selection of pre-chirp parameter values. Our results also demonstrate that the proposed AFDM-PIM scheme is superior to classical AFDM and IM-aided OFDM algorithms in terms of BER performance.
\end{itemize}
	
	The rest of the paper is organized as follows. In Section~\ref{S2}, the AFDM system model is introduced. Section~\ref{S3} details the proposed AFDM-PIM scheme. The performance analysis of AFDM-PIM under doubly dispersive channels is presented in Section~\ref{S4}, which includes the PEP and diversity analysis. The pre-chirp parameter optimization is provided in Section~\ref{S5}. The simulation results and discussions are offered in Section~\ref{S6}, and Section~\ref{S7} draws the conclusions. The notations adopted in this paper are listed in Table~\ref{Notation}.
	
%\textit{Notation:} $\left\lfloor \cdot \right\rfloor$ denotes the integer floor operator. $s\! \sim\! \mathcal{CN}(0,\sigma^{2})$ means that the random variable $s$ follows a complex Gaussian distribution with zero mean and variance $\sigma^{2}$. $x^{*}$ is the conjugate of the complex number $x$. $a\! \mid\! b$ represents that $b$ is divisible by $a$ without leaving a remainder. ${\max\left(a,b \right)}$ and ${\min\left(a,b \right)}$ represent the maximum and minimum values of $a$ and $b$, respectively. $\binom{a}{b} $ denotes the number of ways to choose $b$ elements from a set of $a$ elements. $(\cdot)_{N}$ denotes the modulo $N$ operation. $\mathbf{X}^\mathrm{T}$ and $\mathbf{X}^\mathrm{H}$ stand for the transpose and Hermitian operations of $\mathbf{X}$, respectively. $\left\| \mathbf{X} \right\|_{\mathrm{F}}$ represents the Frobenius norm of $\mathbf{X}$. $\mathbf{I}_N$ denotes the $N\times N$ identity matrix.
%	$Q\left(\cdot \right)$ and $\mathrm{E}(\cdot)$ denote the tail distribution function of the standard Gaussian distribution and the expectation operator, respectively. $\mathbb{I}$ represents the set of irrational numbers. $\Re(x)$ denotes the real part of the complex number $x$, and $\textsf{j}\! =\! \sqrt{-1}$ denotes the imaginary axis. \textcolor{blue}{$\prod$ stands for the product operator, representing the multiplication of a sequence of terms. $\textnormal{Tr}[\cdot] $ denotes the trace operator. All the symbols are listed in Table~\ref{Notation}.} 
	
\begin{table}[!tbp]
\centering
\caption{Mathematical Notations}
\label{Notation} % Tabl.I
\vspace*{-1mm}
\begin{tabular}{|m{2cm}<{\centering}|m{5.8cm}<{\centering}|}
			\hline
			\textbf{Notation} & \textbf{Description} \\ \hline
			$\left\lfloor \cdot \right\rfloor$ & Integer floor operator. \\ \hline
			$s \sim \mathcal{CN}(0,\sigma^{2})$ & Random variable $s$ follows a complex Gaussian distribution with zero mean and variance $\sigma^{2}$. \\ \hline
			$x^{*}$ & Conjugate of the complex number $x$. \\ \hline
			$a \mid b$ & $b$ is divisible by $a$ without leaving a remainder. \\ \hline
			$\max\left(a,b\right)$ & Maximum value of $a$ and $b$. \\ \hline
			$\min\left(a,b\right)$ & Minimum value of $a$ and $b$. \\ \hline
			$\binom{a}{b}$ & Number of ways to choose $b$ elements from a set of $a$ elements. \\ \hline
			$(\cdot)_{N}$ & Modulo $N$ operation. \\ \hline
			$\mathbf{X}^\mathrm{T}$, $\mathbf{X}^\mathrm{H}$ & Transpose and Hermitian operations of $\mathbf{X}$, respectively. \\ \hline
			$\left\| \mathbf{X} \right\|_{\mathrm{F}}$ & Frobenius norm of $\mathbf{X}$. \\ \hline
			$\mathbf{I}_N$ & $N \times N$ identity matrix. \\ \hline
			$Q(\cdot)$ & Tail distribution function of the standard Gaussian distribution. \\ \hline
			$\mathrm{E}(\cdot)$ & Expectation operator. \\ \hline
			$\mathbb{I}$ & Set of irrational numbers. \\ \hline
			$\Re(x)$ & Real part of the complex number $x$. \\ \hline
			$\textsf{j} = \sqrt{-1}$ & Imaginary unit (imaginary axis). \\ \hline
			$\prod$ & Product operator, representing the multiplication of a sequence of terms. \\ \hline
			$\textnormal{Tr}[\cdot]$ & Trace operator. \\ \hline
\end{tabular}
\end{table}

	%%%%%%%%%%%%%%%%%%%%%%%%%%%%%%%%%%%%%%%%%%%%%%%%%%%%%%%%%%%%%%%%%%
	%\newpage
\section{AFDM System Model}\label{S2}

	The general system model of AFDM is presented in Fig.~\ref{fig1}. For clarity, we provide a concise review of the fundamental concept of AFDM \cite{bemani2023affine}. The transmitted bit stream is initially mapped onto a symbol vector, denoted as $\mathbf{x}_A\! =\! \left[ x_A[0], x_A[1], \ldots, x_A[N-1]\right]^\mathrm{T}\! \in\! \mathbb{C}^{N \times 1}$, comprising $N$ $M$-ary amplitude/phase modulations (APM) symbols in the DAFT domain. The resulting signals are then converted to time-domain representations with an $N$-point IDAFT, formulated as
\begin{equation}\label{eqIDAFT} % eq.1
	s_A[n] = \frac{1}{\sqrt{N}} \sum_{m=0}^{N-1} x_A[m] e^{\textsf{j} 2 \pi\left(c_1 n^2 + c_2 m^2 + \frac{m n}{N} \right)},
\end{equation}
where $s_A[n]$ is the time domain signal, and $m,n\! \in\! \{0, 1, \ldots, N-1\}$, while $c_1$ and $c_2$ are the post-chirp and pre-chirp parameters of the DAFT, respectively. 
	
	Similarly to OFDM, AFDM also necessitates the insertion of prefix to address the multi-path problem. By leveraging the inherent periodicity characteristic of DAFT, a chirp-periodic prefix (CPP) is inserted to serve the function analogous to the cyclic prefix (CP) in OFDM, which is given by
\begin{equation}\label{eqCPP} % eq.2
	s_A[n]\! =\! s_A[N\! +\! n]e^{-\textsf{j} 2\pi c_{1}(N^{2}+2Nn)},\, n\! =\! -L_{\mathrm{cp}},\ldots,-1,\!
\end{equation} 
where $L_{\mathrm{cp}}$ is the length of the CPP.

\begin{figure}[!t] %fig.1
%	\vspace{-1mm}
	\centering
	\includegraphics[width=3.4in]{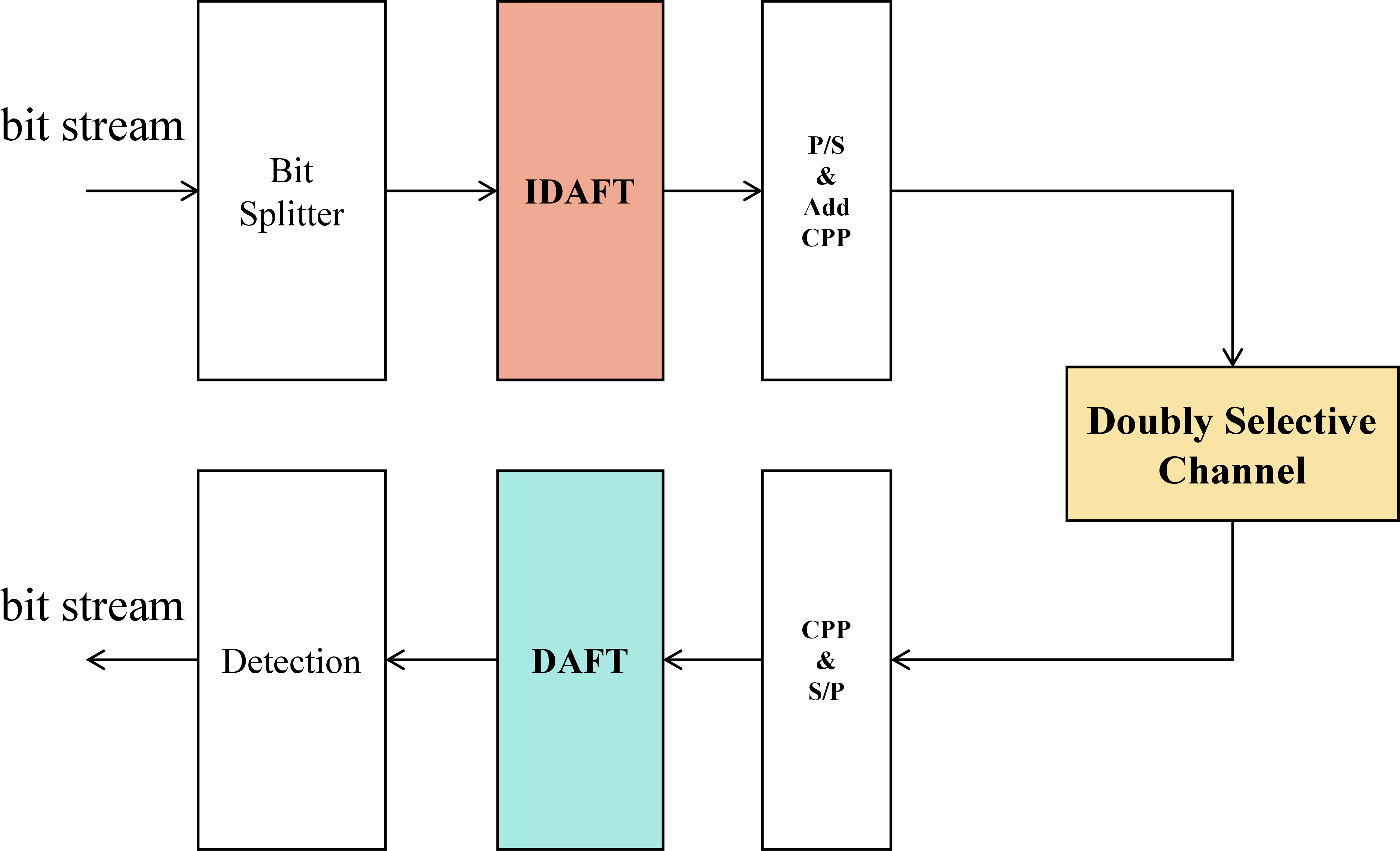}
	\vspace{-2mm}
	\caption{The block diagram of AFDM system.}
	\label{fig1}
\vspace{-2mm}
\end{figure}

	Under high-mobility scenarios, the transmitted signals may experience time-frequency doubly dispersive channel attributed to the severe Doppler shift and multi-path effects, which can be modeled as
\begin{equation}\label{eqCIR}  % eq.3
	h(\tau, \nu) = \sum_{p=1}^P h_p \delta(\tau-\tau_p) \delta(\nu-\nu_p),
\end{equation}
where $P$ is the number of the paths, $\nu_p$ and $\tau_p$ are the Doppler shift and delay of the $p$-th path, respectively, while $h_p\! \sim\! \mathcal{CN}(0,1/P)$ is the $p$-th path's channel coefficient. The normalized delay and Doppler shift are given by $d_p\! =\! \tau_p \Delta f$ and $\alpha_{p}\! =\! N T \nu_{p}$, respectively, where $\Delta f$ is the AFDM subcarrier spacing and $T$ is the sampling interval with $T \Delta f\! =\! 1$. Furthermore, $\alpha_{p}\! \in\! [-\alpha_{\max},\, \alpha_{\max}]$ and $d_p\! \in\! [0,\, d_{\max}] $, where $\alpha_{\max}$ and $d_{\max}$ denote the maximum Doppler shift and maximum delay, respectively \cite{10570960}. For simplicity and without loss of generality, we mainly consider integer values of $\alpha_p$ in this paper. 

	At the receiver, by discarding the CPP, the received time domain symbols can be expressed as
\begin{equation}\label{time-signals} % eq.4
	r_A[n] = \sum_{p=1}^{P} h_p s_A[n - d_p] e^{-\textsf{j} 2\pi \nu_p n} + w[{n}] ,
\end{equation}
where $w[{n}]\! \sim\! \mathcal{C N}\left(0, N_0\right)$ is the additive white Gaussian noise (AWGN). After the $N$-point DAFT, the received AFDM signals in the DAFT domain can be expressed as 
\begin{align}\label{eqRxS} % eq.5                     
	y_A[\bar{m}] = \frac{1}{\sqrt{N}}&\sum_{n=0}^{N-1} r_A[n] e^{-\textsf{j} 2 \pi\left(c_1 n^2+c_2 \bar{m}^2+n \bar{m} / N\right)} \nonumber \\
	& \bar{m} =0, 1, \ldots, N-1,
\end{align}
where $\bar{m}$ are the indices in the DAFT domain. In matrix form, the received AFDM signals can be further written as
\begin{equation}\label{eqRxSv}  % eq.6
	\mathbf{y}_A = \sum_{p=1}^P h_p \mathbf{A} \mathbf{\Gamma}_{\textnormal{CPP}_p} \bm{\Delta}_{\nu_p} \bm{\Pi}^{d_p} \mathbf{A}^\mathrm{H} \mathbf{x}_A + \mathbf{w} ,
\end{equation}   
where the diagonal matrix $\bm{\Delta}_{\nu_p}\! =\! \operatorname{diag}\big(e^{-\textsf{j} 2 \pi \nu_p \cdot 0},e^{-\textsf{j} 2 \pi \nu_p \cdot 1},$ $\ldots, e^{-\textsf{j} 2 \pi \nu_p \cdot (N-1)} \big)$ represents the Doppler effect, $\mathbf{A}$ is the DAFT matrix, and $\mathbf{w}$ is the noise vector, while $\bm{\Pi}$ represents the forward cyclic-shift matrix given by 
\begin{equation}\label{eqCSM}  % eq.7
	\mathbf{\Pi} = \left[\begin{array}{cccc}0&\cdots&0&1\\1&\cdots&0&0\\\vdots&\ddots&\ddots&\vdots\\0&\cdots&1&0\end{array}\right]_{N\times N},
\end{equation}
and $\mathbf{\Gamma}_{\text{CPP}_p}=\operatorname{diag}\big(\omega_{p,0}, \omega_{p,1}, \ldots, \omega_{p,N} \big)$ is the $N \times N$ diagonal matrix for CPP with
\begin{equation}\label{CPP_matrix} % eq.8
	\omega_{p,n} = \begin{cases}e^{-\textsf{j} 2\pi c_1(N^2-2N(d_p-n))},& n<d_p, \\ \qquad \qquad 1, & n\ge d_p . \end{cases}
\end{equation}
Upon receiving the signal $\mathbf{y}_A$, the ML detector can be employed for signal detection.
	
	%%%%%%%%%%%%%%%%%%%%%%%%%%%%%%%%%%%%%%%%%%%%%%%%%%%%%%%%%%%%%%%%%%
	%\newpage

\begin{figure*}[!b]   % fig.2
\vspace*{-4mm}
\centering
\includegraphics[width=6.4in]{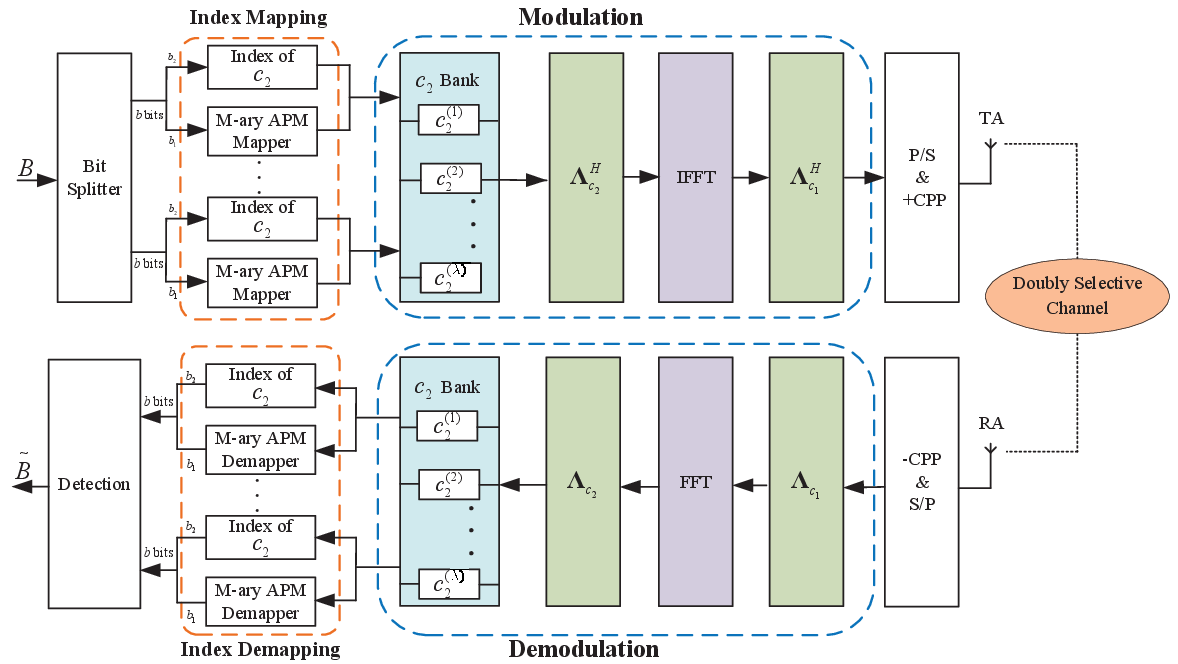}
\vspace*{-3mm}
\caption{Transceiver structure of the proposed AFDM-PIM scheme.}
\label{System-Model}
\vspace*{-1mm}
\end{figure*}	
	
\section{Proposed AFDM-PIM Scheme}\label{S3}

We first analyze the orthogonality of AFDM subcarriers and then derive the proposed AFDM-PIM framework. The input-output relation and parameter settings of the AFDM-PIM are also presented. 

\subsection{Orthogonality Analysis of {AFDM} Subcarriers}\label{S3.1}

	Following the modulation process of AFDM, (1) can also be expressed as 
\begin{equation}\label{eqMod} % eq.9
	s_A[n] = \sum_{m=0}^{N-1} x_A[m] \phi_n(m) ,~ n=0, 1, \ldots,N-1,
\end{equation}
where $\phi_n(m) $ denotes the $m$-th chirp-like subcarrier given by
\begin{equation}\label{eqC-Lsub} % eq.10
	\phi_n(m) = \frac{1}{\sqrt{N}} \, e^{\textsf{j} 2 \pi\left(c_1 n^2+c_2 m^2+\frac{m n}{N}\right)}. 
\end{equation}
	
	The following Theorem demonstrates the flexibility of the pre-chirp parameter $c_2$ assignment for different subcarriers.

\begin{theorem}\label{T1}
Applying distinct $c_2$ to different subcarriers in the AFDM modulation process will still preserve their orthogonality.
\end{theorem}
	
\begin{proof}
  The inner product between two subcarriers of AFDM, which utilize the same post-chirp parameter $c_1$ but distinct values of pre-chirp parameter $c_2$, designated as $\phi_n^{c_1,c_{2,1}}(m)$ and $\phi_n^{c_1,c_{2,2}}(m)$, respectively, is given by 
\begin{align}\label{eqInPr} % eq.11
	& \sum_{n=0}^{N-1} \phi_n^{c_1,c_{2,1}}(m_1) \big(\phi_n^{c_1,c_{2,2}}(m_2)\big)^{*} \nonumber \\
	& \hspace*{4mm}= \frac{1}{N} e^{-\textsf{j} 2\pi\left(c_{2,1}m_1^2-c_{2,2}m_2^2\right)} \sum_{n=0}^{N-1}e^{-\textsf{j} \frac{2\pi}N(m_1-m_2)n} \nonumber \\
	& \hspace*{4mm}= \frac{1}{N} e^{-\textsf{j} 2\pi\left(c_{2,1}m_1^2-c_{2,2}m_2^2\right)} \frac{1-e^{-\textsf{j} 2\pi N\left(\frac{m_1-m_2}N\right)}}{1-e^{-\textsf{j} 2\pi\left(\frac{m_1-m_2}N\right)}} \nonumber \\
	& \hspace*{4mm}= \begin{cases}1,&m_1 = m_2, \\0,&\text{otherwise.}\end{cases}
\end{align}
It is evident that the orthogonality among AFDM subcarriers is maintained when different values of $c_2$ are employed.
\end{proof}

	Inspired by Theorem~\ref{T1}, we design the AFDM-PIM scheme, as shown in Fig.~\ref{System-Model}, which utilizes the flexibility of $c_2$ assignment to convey additional information bits.
	
\begin{table}[!t]
\vspace*{-1mm}
\centering
\caption{Mapping rule between the index bits and the PCPs in the Case of $N_c=4$ and $\lambda=4$.}
\label{Tab1} % Tab.II
\vspace*{-2mm}
	\begin{tabular}{|c||c|c|c|c|}
			\hline
			\multirow{2}{*}{Index bits} & \multicolumn{4}{c|}{PCPs for Each Group} \\
			\cline{2-5}
			& subcarrier 1 & subcarrier 2 & subcarrier 3 & subcarrier 4 \\
			\hline
			0000 & $c_2^{(1)}$ & $c_2^{(2)}$ & $c_2^{(3)}$ & $c_2^{(4)}$ \\
			0001 & $c_2^{(1)}$ & $c_2^{(2)}$ & $c_2^{(4)}$ & $c_2^{(3)}$ \\
			0010 & $c_2^{(1)}$ & $c_2^{(3)}$ & $c_2^{(2)}$ & $c_2^{(4)}$ \\
			0011 & $c_2^{(1)}$ & $c_2^{(3)}$ & $c_2^{(4)}$ & $c_2^{(2)}$ \\
			0100 & $c_2^{(1)}$ & $c_2^{(4)}$ & $c_2^{(2)}$ & $c_2^{(3)}$ \\
			0101 & $c_2^{(1)}$ & $c_2^{(4)}$ & $c_2^{(3)}$ & $c_2^{(2)}$ \\
			0110 & $c_2^{(2)}$ & $c_2^{(1)}$ & $c_2^{(3)}$ & $c_2^{(4)}$ \\
			0111 & $c_2^{(2)}$ & $c_2^{(1)}$ & $c_2^{(4)}$ & $c_2^{(3)}$ \\
			1000 & $c_2^{(2)}$ & $c_2^{(3)}$ & $c_2^{(1)}$ & $c_2^{(4)}$ \\
			1001 & $c_2^{(2)}$ & $c_2^{(3)}$ & $c_2^{(4)}$ & $c_2^{(1)}$ \\
			1010 & $c_2^{(2)}$ & $c_2^{(4)}$ & $c_2^{(1)}$ & $c_2^{(3)}$ \\
			1011 & $c_2^{(2)}$ & $c_2^{(4)}$ & $c_2^{(3)}$ & $c_2^{(1)}$ \\
			1100 & $c_2^{(3)}$ & $c_2^{(1)}$ & $c_2^{(2)}$ & $c_2^{(4)}$ \\
			1101 & $c_2^{(3)}$ & $c_2^{(1)}$ & $c_2^{(4)}$ & $c_2^{(2)}$ \\
			1110 & $c_2^{(3)}$ & $c_2^{(2)}$ & $c_2^{(1)}$ & $c_2^{(4)}$ \\
			1111 & $c_2^{(3)}$ & $c_2^{(2)}$ & $c_2^{(4)}$ & $c_2^{(1)}$ \\
			\hline
	\end{tabular}
	\vspace*{-4mm}
\end{table}
	
\subsection{Transmitter}\label{S3.2}

  Consider the same AFDM symbol vector comprising $N$ $M$-ary constellation symbols in the DAFT domain as in Section~\ref{S2}. At the transmitter, unlike classical AFDM, the $N$ AFDM subcarriers are divided into $G$ groups, with each group comprising $N_c\! =\! N / G$ chirp subcarriers. The total $B$ information bits are split into $G$ parallel streams of $b\! =\! B/G$ bits for each subcarrier group. Each $b$-bit stream is further segmented into $b_1$ symbol bits and $b_2$ index bits, i.e., $b\! =\! b_{1} + b_{2}$. Within the $g$-th group ($1\le g\le G$), the $b_{1}\! =\! N_c\log_{2}(M)$ symbol bits are conveyed by $N_c$ $M$-ary symbols, denoted as $\mathbf{x}^g\! =\! [x^g[0], x^g[1], \ldots, x^g[N_c - 1]]^\mathrm{T}\! \in\! \mathbb{C}^{N_c\times1}$. On the other hand, each subcarrier is assigned with a unique $c_2$ value from a finite alphabet of $\lambda$ legitimate $c_2$ realizations (corresponding to the `$c_2$ Bank' in Fig.~\ref{System-Model}), i.e., $\mathcal{P}_c\! =\! \left\{{c_2^{(1)}, c_2^{(2)}, \ldots, c_2^{(\lambda )}} \right\}$. Specifically, the pre-chirping pattern (PCP) of the $c_2$ values in the $g$-th group, denoted by $\mathrm{P}_{c_2}^g\! =\! [c_{2,N_c\left(g-1\right)}, c_{2,N_c\left(g-1\right)+1}, \ldots, c_{2, N_c g - 1}]^\mathrm{T}\! \in\! \mathbb{C}^{N_c\times1}$ is determined by the $b_2$ index bits according to the pre-defined relationship between the index-bit stream and the permutations of $N_c$ elements from $\mathcal{P}_c$, where $c_{2,m}\! \in \!\mathcal{P}_c$ represents the pre-chirp parameter of the $m$-th subcarrier. $\mathrm{P}_{c_2}\! =\! \left[\mathrm{P}_{c_2}^1, \mathrm{P}_{c_2}^2, \ldots, \mathrm{P}_{c_2}^G \right] $ represents the PCP for all the groups (PCPG), and all the possible PCPGs are combined into a set, termed $\mathcal{S}_p$. Table~\ref{Tab1} exemplifies the mapping rule in the case of $N_c=4$ and $\lambda=4$. Every $b_2$ index bits correspond to one row in Table~\ref{Tab1}. Hence, aside from the APM bits, additional $b_2$ information bits can be implicitly conveyed by the indices of $\mathrm{P}_{c_2}^g$, and
\begin{equation}\label{eqB2bits}  % eq.12
	b_2 = \begin{cases}\left\lfloor\log_2(\mathcal{C}_{\lambda,N_c} N_c!)\right\rfloor, \vspace{0.1cm} & \lambda \geq N_c, \\
	\left\lfloor \log_2(\lambda!) \right\rfloor \frac{N_c}{\lambda}, & \lambda <N_c \text{ \& } \lambda \mid N_c, \\
	\left\lfloor\log_2(\mathcal{C}_{\lambda,N_c} \lambda! \lambda ^{(N_c-\lambda )})\right\rfloor, &  \text{otherwise,} \end{cases}
\end{equation} 
where $\mathcal{C}_{\lambda,N_c}$ is defined as
\begin{equation}\label{eqCl} % eq.13
	\mathcal{C}_{\lambda,N_c} = \binom{\max\left(\lambda,N_c \right)}{\min\left(\lambda,N_c \right)} .
\end{equation} 
For ease of design and subsequent optimization, we typically select the case of \(\lambda = N_c \), which implies \(b_2 = \left\lfloor \log_2(N_c!) \right\rfloor\).
	
	After the mapping for all the groups, the time-domain transmitted signals can be generated through $N$-point IDAFT operation, expressed as
\begin{equation}\label{IDAFT} % eq.14
	s[n] = \frac{1}{\sqrt{N}} \sum_{m=0}^{N-1} x[m] e^{\textsf{j} 2 \pi\left(c_1 	n^2+c_{2,m} m^2+n m / N\right)}. 
\end{equation}
The matrix form of (\ref{IDAFT}) can be formulated as
\begin{equation} % eq.15
	\mathbf{s} = \mathbf{A}^\mathrm{H} \mathbf{x} = \mathbf{\Lambda}_{c_1}^\mathrm{H} \mathbf{F}^\mathrm{H} \bm{\Lambda}_{c_2}^\mathrm{H} \mathbf{x},
\end{equation}
where $\mathbf{x}\! =\! \left[\big(\mathbf{x}^1\big)^{\rm T}, \big(\mathbf{x}^2\big)^{\rm T}, \ldots, \big(\mathbf{x}^G\big)^{\rm T}\right]^\mathrm{T} $, and $\bm\Lambda_{c_2}$ and $\bm\Lambda_{c_1}$ represent the pre-chirp and post-chirp diagonal matrices, respectively, expressed as
\begin{equation}\label{eq-c2} % eq.16
	\bm\Lambda_{c_2} = \operatorname{diag}\left(e^{-\textsf{j} 2\pi c_{2,m} m^2},m=0,1,\ldots,N-1\right),
\end{equation}
\begin{equation}\label{eq-c1} % eq.17
	\bm\Lambda_{c_1} = \operatorname{diag}\left(e^{-\textsf{j} 2\pi c_1 n^2}, n=0,1,\ldots,N-1\right) ,
\end{equation}
while $\mathbf{F}$ denotes the DFT matrix with elements $\mathbf{F}(m,n) = e^{-\textsf{j} 2 \pi m n/N}/\sqrt{N}$, $m,n=0,1,\ldots, N-1$.
	
	Like classical AFDM, our proposed AFDM-PIM also requires the CPP to mitigate the effects of multi-path propagation effectively. Without loss of generality, the length of the CPP is assumed to be greater than the maximum channel delay spread. 

\subsection{Receiver}\label{S3.3}
	
	At the receiver, considering the doubly dispersive channel in (\ref{eqCIR}), the received time-domain signals after removing the CPP can be written as 
\begin{equation}\label{received-signals} % eq.18
	r[n] = \sum_{p=1}^{P} h_p s[n - d_p] e^{-\textsf{j} 2\pi \nu_p n} + w_r[n] ,
\end{equation}
where $w_r[n] \sim \mathcal{CN}(0,N_{0})$ is the AWGN. The matrix form of (\ref{received-signals}) is given by
\begin{equation}\label{eqMRxS} % eq.19
	\mathbf{r} = \mathbf{H}\mathbf{s} + \mathbf{w} = \sum_{p=1}^{P} h_{p} \mathbf{\Gamma}_{\textnormal{CPP}_{p}} \bm{\Delta}_{\nu_{p}} \mathbf{\Pi}^{d_{p}} \mathbf{s} + \mathbf{w}_r,
\end{equation}    
where $\mathbf{w}_r\! =\! [w_r[0], w_r[1],\ldots,w_r[N-1]]^\mathrm{T}$ is the time-domain noise vector. As given in Section~\ref{S2}, $\bm{\Delta}_{\nu_{p}}$ represents the Doppler effect, $\mathbf{\Pi}$ is the forward cyclic-shift matrix with $\mathbf{\Pi}^{d_{p}}$ modeling the delay extension, and $\mathbf{\Gamma}_{\mathrm{CPP}_{p}}$ is the effective CPP matrix.
	
	By applying the DAFT operation, the received DAFT-domain symbols are obtained as
\begin{equation}\label{eqY} % eq.20
	y[\bar{m}]=\frac{1}{\sqrt{N}}\sum_{n=0}^{N-1} r[n] e^{-\textsf{j} 2 \pi\left(c_1 n^2+c_{2,\bar{m}} \bar{m}^2+n \bar{m} / N\right)},
\end{equation}
which can also be grouped into a vector as
\begin{align}\label{y_matrix} % eq.21
	\mathbf{y} =& \mathbf{A} \mathbf{r} = \sum_{p=1}^P h_p \mathbf{A}\mathbf{\Gamma}_{\mathrm{CPP}_p} \mathbf{\Delta}_{\nu_p} \mathbf{\Pi}^{d_p}\mathbf{A}^\mathrm{H} \mathbf{x} + \mathbf{A}\mathbf{w}_r \nonumber \\
	=& \mathbf{H}_\textnormal{eff}\mathbf{x} + \mathbf{w},
\end{align}
where $\mathbf{H}_\textnormal{eff}$ is the effective channel matrix in the DAFT-domain and $\mathbf{w} = \mathbf{A}\mathbf{w}_r$. 

The time-frequency doubly dispersive channel can be estimated through pilot-aided channel estimation algorithms \cite{10557524,8671740,10711268}. This paper will not provide further elaboration on this subject for brevity. Given the estimated $\mathbf{\hat{H}}_\mathrm{eff}$ for the effective channel matrix $\mathbf{H}_\textnormal{eff}$, the ML data detection is formulated as the following optimization
\begin{equation}\label{ML1} % eq.22
	\left(\hat{\mathbf{x}},\hat{\mathrm{P}}_{c_2} \right)= \arg \min\limits_{\forall \mathbf{x},\mathrm{P}_{c_2}}\left\|\mathbf{y}-\mathbf{\hat{H}}_\mathrm{eff} \mathbf{x}\right\|^2 .
\end{equation}

\begin{figure*}[!bp]\setcounter{equation}{31}
\vspace*{-4mm}
\hrulefill
\begin{align}\label{eqCPEP}  % eq.32
	& \Pr([\mathbf{x},\mathbf{\Phi}]\! \to\! [\hat{\mathbf{x}},\hat{\mathbf{\Phi}}]|\mathbf{h}) = \Pr(\|\mathbf{y}\! -\! \hat{\mathbf{\Phi}} (\hat{\mathbf{x}})\mathbf{h}\|^2\! <\! \|\mathbf{y}\! -\! \mathbf{\Phi}(\mathbf{x})\mathbf{h}\|^2) \nonumber \\
	& \hspace*{10mm} = \mathrm{Pr}\left(\textnormal{Tr}\left[(\mathbf{y} - \hat{\mathbf{\Phi}} (\hat{\mathbf{x}})\mathbf{h})^{\rm H} (\mathbf{y} - \hat{\mathbf{\Phi}} (\hat{\mathbf{x}})\mathbf{h}) - (\mathbf{y} - \mathbf{\Phi}(\mathbf{x})\mathbf{h})^{\rm H} (\mathbf{y}-\mathbf{\Phi}(\mathbf{x})\mathbf{h})\right]<0\right)  \nonumber \\
	& \hspace*{10mm} = \mathrm{Pr}\left(\textnormal{Tr}\left[\mathbf{h}^{\rm H}\left(\mathbf{\Phi}(\mathbf{x}) - \hat{\mathbf{\Phi}} (\hat{\mathbf{x}})^{\rm H}\right)\left(\mathbf{\Phi}(\mathbf{x}) - \hat{\mathbf{\Phi}} (\hat{\mathbf{x}})^{\rm H}\right)\mathbf{h}\right]-\mathrm{C}<0\right) 
   = \Pr\left(\mathbf{\chi}>\big\|\big(\hat{\mathbf{\Phi}} (\hat{\mathbf{x}}) - \mathbf{\Phi}(\mathbf{x})\big)\mathbf{h}\big\|^2 \right),
\end{align} 
\vspace*{-1mm}
\end{figure*} 
	
\subsection{Input-Output Relation and Parameter Settings}\label{S3.4}

	Substituting (\ref{IDAFT}) and (\ref{received-signals}) into (\ref{eqY}), the input-output relation of AFDM-PIM can be obtained as\setcounter{equation}{22}
\begin{equation}\label{input-output} % eq.23
		y[\bar{m}] \! = \! \frac{1}{N} \sum_{p=1}^P \! \sum_{m=0}^{N-1} \! h_p \xi_{(p, \bar{m}, m) } \eta_{(p, \bar{m}, m)} x[m] + w[\bar{m}],
\end{equation}
where 
\begin{equation}\label{eqXI} % eq.24
	\xi_{(p, \bar{m}, m)} = e^{\textsf{j}\frac{2\pi}{N}\left(N c_{2,m}m^2-N c_{2,\bar{m}}\bar{m}^2-m d_p+N c_1 d_p^2\right)},
\end{equation}
\begin{align}\label{eqETA} % eq.25
	\eta_{(p, \bar{m}, m)} =& \sum_{n=0}^{N - 1} e^{-\textsf{j} \frac{2\pi}{N}((\bar{m} - m + \alpha_{p} + 2N c_1 d_p)n)} \nonumber \\
	=& \frac{e^{-\textsf{j} 2\pi(\bar{m} - m + \alpha_{p}+2N c_{1}d_{p})}-1}{e^{-\textsf{j} \frac{2\pi}{N}(\bar{m} - m+\alpha_{p}+2N c_{1}d_{p})}-1}.
\end{align}
In matrix representation, (\ref{input-output}) can be rewritten as
\begin{equation} \label{input-output_matrix}% eq.26
	\mathbf{y} = \sum_{p=1}^P h_p \mathbf{H}_p \mathbf{x} + \mathbf{w},
\end{equation}
where the elements of $\mathbf{H}_p$ are given by 
\begin{align}\label{H_p} % eq.27
	H_p[\bar{m},m] =& \frac{1}{N} \xi_{(p, \bar{m}, m) } \eta_{(p, \bar{m}, m)} \nonumber \\
	=& \begin{cases}\xi_{(p, \bar{m}, m) }, & m=(\bar{m}+\mathrm{loc}_p)_{N}, \\\qquad 0, & \text{otherwise,}\end{cases}
\end{align}
where $\mathrm{loc}_p = (\alpha_{p} + 2 N c_1 d_p)_{N}$. Since the range of $\mathrm{loc}_p$ is $[-\alpha_{\max} + 2 N c_1 d_p,\, \alpha_{\max} + 2 N c_1 d_p]$, we define $\mathrm{loc}_p \in \mathbb{K}_p$, with $\mathbb{K}_{p} = \{ -\alpha_{\mathrm{max}} + 2 N c_1 d_p, \ldots, \alpha_{\mathrm{max}} + 2 N c_1 d_p\}$.

	It can be seen from (\ref{H_p}) that for the two chirp parameters of AFDM-PIM, only the post-chirp parameter exerts an influence on the positions of non-zero entries in the matrix $\mathbf{H}_p$ determined by $\mathrm{loc}_p$, which is independent of the pre-chirp parameter.
	Therefore, like classical AFDM, full diversity order can be obtained by the proposed AFDM-PIM under doubly dispersive channels by adjusting $c_1$  to avoid possible overlap between non-zero elements of $\mathbf{H}_i$ and $\mathbf{H}_j$ $(i\neq j)$. Specifically, this requires that the intersection between the corresponding ranges for $\mathrm{loc}_i$ and $\mathrm{loc}_j$ is empty, i.e.,\setcounter{equation}{27}
\begin{equation}\label{empty} % eq.28
	\mathbb{K}_{i} \cap \mathbb{K}_{j}=\emptyset.
\end{equation}	
Without loss of generality, assume that $d_i \le d_j$. Then the constraint (\ref{empty}) can be transformed into
\begin{equation}\label{eqTcont} % eq.29
		c_1 > \frac{2\alpha_{\max}}{2N(d_j-d_i)}.
\end{equation}
Since the minimum value of $(d_j - d_i)$ with $i\neq j$ equals one, it can be concluded that $c_1$ can be set as:
\begin{equation}\label{c1} % eq.30
	c_1 = \frac{2\alpha_{\max}+1}{2N}.
\end{equation}
Following the configuration (\ref{c1}) and $(d_{\max}+1)(2\alpha_{\max}+1)\leq N$, the channel paths with different delays or Doppler shifts can be distinguished within the DAFT domain, as illustrated in Fig.~\ref{fig3}, which shows an example of the effective channel matrix of an AFDM-PIM system under a three-path channel. The rigorous diversity analysis will be provided in Section~\ref{S4}.

\begin{figure}[!t] %fig.3
\centering
\includegraphics[width=2.5in]{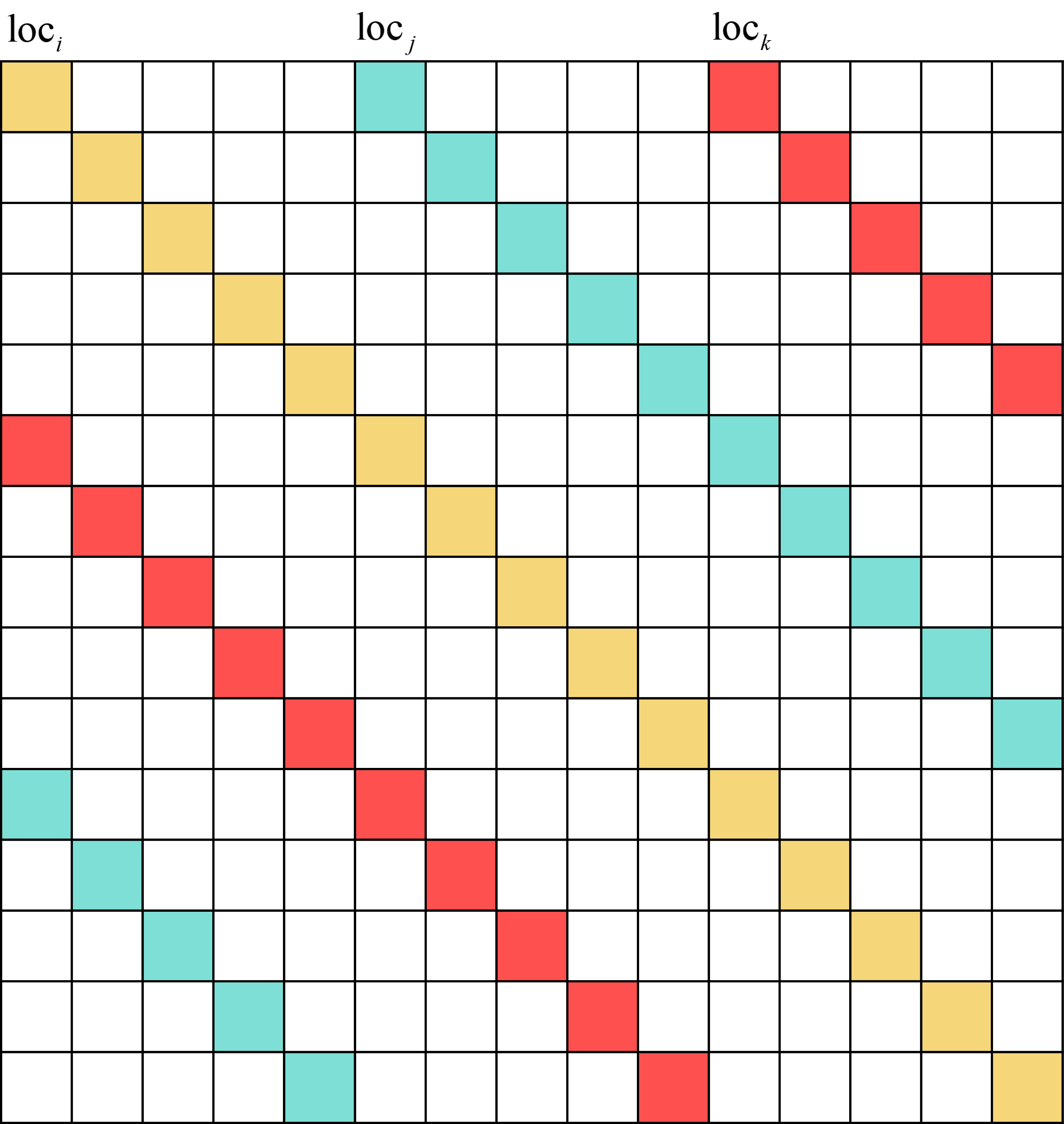}
\vspace{-2mm}
\caption{Example of the effective channel matrix of a three-path channel.}
\label{fig3}
\vspace{-3mm}
\end{figure}
	
	%%%%%%%%%%%%%%%%%%%%%%%%%%%%%%%%%%%%%%%%%%%%%%%%%%%%%%%%%%%%%%%%%%
	% \newpage
\section{Performance Analysis}\label{S4}

	We first derive the average BEP (ABEP) upper bounds for our AFDM-PIM scheme with the ML detection. Then we analyze the diversity order achieved by the system.
	
\subsection{Error Performance Analysis}\label{S4.1}

	To facilitate the analysis, the received DAFT-domain signal (\ref{input-output_matrix}) can be rewritten as
\begin{equation}\label{inputoutput29} % eq.31
	\mathbf{y} = \mathbf{\Phi}(\mathbf{x})\mathbf{h} + \mathbf{w},
\end{equation}
where $\mathbf{\Phi}(\mathbf{x})\! =\! [\mathbf{H}_{1}\mathbf{x},\mathbf{H}_{2}\mathbf{x}, \ldots, \mathbf{H}_{P}\mathbf{x}]\! \in\! \mathbb{C}^{N \times P}$ and $\mathbf{h}\! =\! [h_{1},h_{2},\ldots,h_{P}]^{\rm T}\! \in\! \mathbb{C}^{P \times 1}$. The conditional PEP (CPEP) between $\mathbf{\Phi} (\mathbf{x})$ and its estimate $\hat{\mathbf{\Phi}} (\hat{\mathbf{x}})$ can be calculated as (\ref{eqCPEP}), shown at the bottom of this page,
where $\chi\! =\! \mathbf{w}^{\rm H}\big(\hat{\mathbf{\Phi}} (\hat{\mathbf{x}})\! -\! \mathbf{\Phi}(\mathbf{x})\big)\mathbf{h}\! +\! \mathbf{h}^{\rm H} \big(\hat{\mathbf{\Phi}} (\hat{\mathbf{x}})\! -\! \mathbf{\Phi}(\mathbf{x})\big)^{\rm H} \mathbf{w}$. Since $\chi$ follows Gaussian distribution with variance $2N_{0}\big\|(\hat{\mathbf{\Phi}} (\hat{\mathbf{x}}) - \mathbf{\Phi}(\mathbf{x}))\mathbf{h}\big\|^2$, the CPEP can be expressed as \setcounter{equation}{32}
\begin{equation}\label{eqCPEP1} % eq.33
	\Pr\big([\mathbf{x},\mathbf{\Phi}]\to[\hat{\mathbf{x}},\hat{\mathbf{\Phi}} ]|\mathbf{h}\big) = Q\biggl(\sqrt{\frac{\delta}{2N_{0}}}\biggr),
\end{equation}
where $\delta\! =\! \left\|\big(\hat{\mathbf{\Phi}} (\hat{\mathbf{x}})\! -\! \mathbf{\Phi}(\mathbf{x})\big)\mathbf{h}\right\|^2\! =\!\mathbf{h}^\mathrm{H}\mathbf{\Psi}\mathbf{h}$ with $\mathbf{\Psi}\! =\! \big(\hat{\mathbf{\Phi}} (\hat{\mathbf{x}})\! -\! \mathbf{\Phi}(\mathbf{x})\big)^\mathrm{H}\big(\hat{\mathbf{\Phi}} (\hat{\mathbf{x}})\! -\! \mathbf{\Phi}(\mathbf{x})\big)$. According to \cite{1188428}, $Q$-function can be approximated as 
\begin{equation}\label{eqQF} % eq.34
	Q(x) \approx \frac{1}{12}e^{-x^2/2}+\frac{1}{4}e^{-2x^2/3}.
\end{equation}
Hence the CPEP can be approximated as 
\begin{equation}\label{eqCPEP2} % eq.35
	\Pr\big([\mathbf{x},\mathbf{\Phi}]\to[\hat{\mathbf{x}},\hat{\mathbf{\Phi}} ]|\mathbf{h}\big)\approx \frac1{12}e^{-\varsigma _1\delta} + \frac14e^{-\varsigma _2\delta},
\end{equation}
with $\varsigma_1\! =\! \frac{1}{4 N_0}$ and $\varsigma_2\! =\! \frac{1}{3 N_0}$. Consequently, the unconditional PEP (UPEP) can be calculated as
\begin{align}\label{PEP1} % eq.36
	\Pr\big([\mathbf{x},\mathbf{\Phi}] & \to [\hat{\mathbf{x}},\hat{\mathbf{\Phi}} ]\big) =\mathrm{E}\Big(\Pr\big([\mathbf{x},\mathbf{\Phi}]\to[\hat{\mathbf{x}},\hat{\mathbf{\Phi}} ]|\mathbf{h}\big)\Big) \nonumber \\
	&\approx \int_{0}^{+\infty}\biggl(\frac{1}{12}e^{-\varsigma_{1}\delta}+\frac{1}{4}e^{-\varsigma_{2}\delta}\biggr)p_{\delta}(\delta) \mathrm{d} \delta , 
\end{align}
where $p_{\delta}(\delta)$ denotes the probability density function (PDF) of $\delta$. Noting $p_{\delta}(\delta)=0$ for $\delta <0$ and utilizing the definition of the moment-generating function (MGF) $M_{\eta}(s)\! =\! \mathrm{E}\big(e^{s\eta}\big)\! =\! \int_{-\infty}^{+\infty}e^{s\eta}p_{\eta}(\eta )\mathrm{d} \eta$, we can calculate the UPEP (\ref{PEP1}) as 
\begin{equation}\label{PEP2} % eq.37
	\Pr\big([\mathbf{x},\mathbf{\Phi}]\to[\hat{\mathbf{x}},\hat{\mathbf{\Phi}}]\big) \approx \frac{1}{12}M_\delta(-\varsigma_1)+\frac{1}{4}M_\delta(-\varsigma_2).
\end{equation}

\begin{theorem}[\cite{MGF}]\label{T2}
  For an $N \times N$ Hermitian matrix $\mathbf{Q}$ and an $N \times 1$ zero-mean complex-valued random vector $\mathbf{v}$ with covariance matrix $\mathbf{L}$, the characteristic function of the quadratic form $f\! =\! \mathbf{v}^\mathrm{H} \mathbf{Q} \mathbf{v}$ can be expressed as
\begin{equation}\label{eqT2-1} % eq.38
	\varphi_{f}(t) = |\mathbf{I} - \textsf{j} \cdot t \mathbf{L} \mathbf{Q}|^{-1} = \prod_{\iota = 1}^{\kappa} \frac{1}{1 - \textsf{j} \cdot t \lambda_{\iota}^{(\mathbf{L} \mathbf{Q})}},
\end{equation}
where $\lambda_{\iota}^{(\mathbf{L} \mathbf{Q})}$ is the $\iota$-th non-zero eigenvalue of matrix $\mathbf{L} \mathbf{Q}$, and $\kappa$ denotes the number of the non-zero eigenvalues, i.e., the rank of $\mathbf{L} \mathbf{Q}$. Using the relationship $M_f(s) = \varphi_{f}(-\textsf{j}\cdot s)$, the MGF of $f$ is given by
\begin{equation}\label{eqT2-2} % eq.39
	M_f(s) = \prod_{\iota = 1}^{\kappa}\frac{1}{1-s\lambda_{\iota}^{(\mathbf{L} \mathbf{Q})}}.
\end{equation}
\end{theorem}

	Since $\mathbf{h}$ is a zero mean complex vector with the covariance matrix $\frac{1}{P}\mathbf{I}_P$ and $\mathbf{\Psi}$ is a Hermitian matrix, according to Theorem~\ref{T2}, the UPEP (\ref{PEP2}) can be further expressed as
\begin{align}\label{PEP_final} % eq.40
	& \Pr\! \big([\mathbf{x},\mathbf{\Phi}]\! \to\! \big[\hat{\mathbf{x}},\hat{\mathbf{\Phi}}\big]\big)\! \approx\! \frac{1}{12}\prod_{\iota=1}^{\kappa} \frac{1}{1+\frac{\lambda_{\iota}^{(\mathbf{\Psi})}}{4PN_{0}}}\! +\! \frac{1}{4}\prod_{\iota=1}^{\kappa} \frac{1}{1\! +\! \frac{\lambda_{\iota}^{(\mathbf{\Psi})}}{3PN_{0}}},
\end{align}
where we assume that the rank of $\mathbf{\Psi}$ is $\kappa$.
	
	Moreover, based on the UPEP (\ref{PEP_final}), the ABEP upper bound for the proposed AFDM-PIM scheme can be calculated by 
\begin{align}\label{ABEP} % eq.41
	{\Pr}_{\textnormal{ABEP}} \leq & \frac{1}{b\,2^{b}} \sum_{\mathbf{x}} \sum_{\hat{\mathbf{x}}} \sum_{\mathbf{\Phi}} \sum_{\hat{\mathbf{\Phi}}} \Pr\big([\mathbf{x},\mathbf{\Phi}]\to[\hat{\mathbf{x}},\hat{\mathbf{\Phi}} ]\big) \nonumber \\
	&\hspace*{25mm}\times \tau\big([\mathbf{x},\mathbf{\Phi}]\to[\hat{\mathbf{x}},\hat{\mathbf{\Phi}}]\big) ,
\end{align}
where $\tau\big([\mathbf{x},\mathbf{\Phi}]\to[\hat{\mathbf{x}},\hat{\mathbf{\Phi}}]\big)$ represents the number of error bits caused by the corresponding pairwise error event.
	
\subsection{Diversity Analysis}\label{S4.2}

	At the high signal-to-noise ratio (SNR) region, the approximation of (\ref{PEP_final}) can be simplified as
\begin{align}\label{eqDA} % eq.42
	& \Pr\big([\mathbf{x},\mathbf{\Phi}]\to[\hat{\mathbf{x}},\hat{\mathbf{\Phi}}]\big) \approx\! \frac{1}{12}\prod_{\iota=1}^{\kappa} \frac{4PN_{0}}{\lambda_{\iota}^{(\mathbf{\Psi})}}\! +\! \frac{1}{4}\prod_{\iota=1}^{\kappa} \frac{3PN_{0}}{\lambda_{\iota}^{(\mathbf{\Psi})}} \nonumber \\
	&\hspace*{5mm} \approx \left(\prod_{\iota=1}^\kappa\lambda_{\iota}^{(\mathbf{\Psi})}\right)^{-1} \left(\frac{(4P)^\kappa}{12}+\frac{(3P)^\kappa}4\right) \mathrm{SNR}^{-\kappa},
\end{align}
where $\mathrm{SNR}$ is defined as $1/N_0$. As a result, the diversity order $\mu$ of the proposed AFDM-PIM scheme equals to the minimum value of $\kappa$, i.e., 
\begin{equation}\label{eqDA1} % eq.43
	\mu = \min \, \mathrm{rank}(\mathbf{\Psi}).
\end{equation}
Since the eigenvalues of Hermitian matrix $\mathbf{\Psi}$ can be calculated as the square of the singular values of $\big(\hat{\mathbf{\Phi}} (\hat{\mathbf{x}}) - \mathbf{\Phi}(\mathbf{x})\big)$, $\mathrm{rank}(\mathbf{\Psi}) \! =\! \mathrm{rank} \big(\hat{\mathbf{\Phi}} (\hat{\mathbf{x}}) - \mathbf{\Phi}(\mathbf{x})\big)$ and the diversity order $\mu$ can also be expressed as 
\begin{equation}\label{order} % eq.44
	\mu = \min\, \mathrm{rank} \left(\hat{\mathbf{\Phi}} (\hat{\mathbf{x}}) - \mathbf{\Phi}(\mathbf{x})\right).
\end{equation} % eq.44
Defining $\mathbf{\Phi}\big(\bar{\delta}\big)\! =\! \hat{\mathbf{\Phi}}\big(\hat{\mathbf{x}}\big) - \mathbf{\Phi}(\mathbf{x})$, the full diversity order analysis for the AFDM-PIM can be transformed into the full rank analysis of $\mathbf{\Phi}\big(\bar{\delta}\big)$ as stated in Theorem~\ref{T3} below.

\begin{theorem}\label{T3}
The proposed AFDM-PIM scheme is capable of achieving the full diversity order if the following two conditions are met.

{\bf Condition 1}: The number of the paths satisfies
\begin{equation}\label{condition1} % eq.45
	P \le (d_{\mathrm{max}}+1)(2\alpha_{\mathrm{max}}+1) \le N.
\end{equation}

{\bf Condition 2}: The pre-chirp parameters in $\mathcal{P}_c$ take the irrational numbers.
\end{theorem}

\begin{proof}	
See Appendix~\ref{ApA}.
\end{proof}

\subsection{Spectral Efficiency Analysis}\label{S4.3}

This subsection examines the SE of the proposed AFDM-PIM scheme. Classic AFDM and IM-assisted AFDM (referred to as AFDM-IM) are utilized as benchmarks to highlight the superiority of the proposed approach.

The SEs of the classic AFDM \cite{bemani2023affine} and AFDM-IM \cite{10342712} are given by 
\begin{equation} % eq.46
	S_{\mathrm{AFDM}}=\log _2(M)
\end{equation}
and    
\begin{equation} % eq.47
	S_{\mathrm{AFDM}-\mathrm{IM}}=\frac{\left\lfloor\log _2\binom{N_c}{a}\right\rfloor+a \log _2(M)}{N_c} ,
\end{equation}
where $N_c$ and $a$ represent the numbers of subcarriers in each group and activated subcarriers at each transmission, respectively, for the AFDM-IM scheme.

Based on the description in Subsection~\ref{S3.2}, the SE of the proposed AFDM-PIM can be calculated as
\begin{equation} % eq.48
	S_{\mathrm{AFDM}-\mathrm{PIM}}= \frac{\left\lfloor \log_2(N_c!) \right\rfloor}{N_c} + \log _2(M).
\end{equation}

As can be seen from the above analysis, the proposed AFDM-PIM demonstrates a higher SE than AFDM, due to the presence of index modulation. Furthermore, in contrast to AFDM-IM, which employs inactive subcarriers at each transmission, AFDM-PIM exhibits a higher SE, attributed to the capacity of all subcarriers to carry modulated data. In summary, compared with classic methods, the proposed AFDM-PIM scheme offers an enhancement in terms of SE, without incurring any additional hardware cost.

\section{Parameter Optimization}\label{S5}
	
	In this section, we establish an optimization problem for the optimal $c_2$ alphabet design, where the PSO algorithm is employed to enhance the BER performance.
	
\subsection{Problem Formulation}\label{S5.1}

	For tractable analysis, the pre-chirp alphabet design is investigated based on the optimal BER detector \cite{7936676}, i.e., the ML detector based on the signal model (\ref{inputoutput29}), expressed as
\begin{equation}\label{eqOPorig} % eq.49  
	\big(\hat{\mathbf{x}},\hat{\mathrm{P}}_{c_2}\big) =  \arg \min_{\forall \mathbf{x},\mathrm{P}_{c_2}}\|\mathbf{y}-\mathbf{\Phi}(\mathbf{x})\mathbf{h}\|^2.
\end{equation}
To achieve the optimal $c_2$ alphabet, the minimum Euclidean distance (MED) between different realizations of $\mathbf{\Phi}(\mathbf{x})$ should be maximized. The Euclidean distance between two realizations of $\mathbf{\Phi}(\mathbf{x})$ is formulated as
\begin{equation}\label{Okj_o}  % eq.50
	O_{k,j}(\mathcal{P}_c) = \sum_{\mathbf{x}^{\prime}\!, \mathbf{x}}\sum_{r = 1 }^{\mathcal{R}}\left\|(\mathbf{\Phi}_k^r (\mathbf{x}^{\prime})-\mathbf{\Phi}_j^r(\mathbf{x}))\right\|_{\mathrm{F}}^2,
\end{equation}
where $j$ and $k$ represent the indices of $\mathcal{S}_{p}$, and the corresponding PCPGs are expressed as $\mathrm{P}_{c_2}^{(j)}\! =\! [c_{2,1}, c_{2,2},\ldots, c_{2, N}]$ and $\mathrm{P}_{c_2}^{(k)}\! = \![c_{2,1}^{\prime}, c_{2,2}^{\prime}, \ldots, c_{2, N}^{\prime}]$, respectively, while $r$ is the index of delay and Doppler selection. Specifically, each $r$ corresponds to a specific combination of $P$ paths with different delays and Doppler shifts under doubly dispersive channels with a maximum delay $d_{\text{max}}$ and a normalized Doppler shift $\alpha_{\text{max}}$, and the maximum value of $r$ is $\mathcal{R} = \binom{P_{\mathrm{max}}}{P}$, where $P_{\mathrm{max}}\! =\! (d_{\mathrm{max}}+1)(2\alpha_{\mathrm{max}}+1)$.
	
	According to Theorem~\ref{T3} and the periodicity brought by $2\pi$ in (\ref{eq-c2}), $c_2$ is an irrational number with a principal value range of [0,\, 1]. To enhance the BER performance of the proposed AFDM-PIM, we formulate the following problem to maximize the MED with optimal $c_2$ alphabet design: 
\begin{subequations}\label{problem_origin} %eq.51a - eq.51.d
	\begin{align}
		& \max_{\mathcal{P}_c} \min_{k,j \in [1,2^{b_2}]} O_{k,j}(\mathcal{P}_c) , \label{problem1} \\
		& \text{ s.t. }\quad j \ne k, \label{constraint1} \\
		&\hspace{.6cm} \quad c_2 \in [0,1], \label{constraint2} \\
		&\hspace{.6cm} \quad c_2 \in \mathbb{I} \label{constraint3}.
	\end{align}
\end{subequations}	
To derive the objective function, consider the elements of $\mathbf{x}^{\prime}$ and $\mathbf{x}$, $x_{\mathrm{loc}_p+n}^{\prime}$ and $x_{\mathrm{loc}_p+n}$, which can be expressed as $x_{\mathrm{loc}_p+n}^{\prime}\! =\! r_{p,n} e^{\textsf{j}\vartheta_{p,n}}$ and $x_{\mathrm{loc}_p+n}\! =\! r_{p,n}^{\prime} e^{\textsf{j}\vartheta_{p,n}^{\prime}}$. Define
\begin{equation}\label{theta_define} % eq.52
	\left\{\begin{array}{l}
		\theta_n = 2 \pi\left(c_{2,\operatorname{loc}_p+n}\left(\operatorname{loc}_p+n\right)^2-c_{2, n} n^2\right) , \\
		\theta_n^{\prime} = 2 \pi\left(c_{2,\operatorname{loc}_p+n}^{\prime}\left(\operatorname{loc}_p+n\right)^2-c_{2, n}^{\prime} n^2\right).
	\end{array}\right.
\end{equation}
and
\begin{align}\label{eq-Fn} % eq.53
  & F_{p,n} \!=\! r_{p,n}^2 \!+\! {r_{p,n}^{\prime \; 2}} \! - 2 r_{p,n}r_{p,n}^{\prime} \cos \left(\theta_n^{\prime} \!-\! \theta_n + \psi_{p,n} \right) ,
\end{align}
where $\psi_{p,n}\! =\! \vartheta_{p,n}^{\prime}\! -\! \vartheta_{p,n}$. Then the objective function $O_{k,j}(\mathcal{P}_c)$ can be expressed as
\begin{equation}\label{Okj} % eq.54
	O_{k,j}(\mathcal{P}_c) \! = \! \sum_{\mathbf{x}^{\prime}\!, \mathbf{x}} \sum_{r=1}^{\mathcal{R}} \sum_{p=1}^P \! \sum_{n=0}^{N-1} F_{p,n}.
\end{equation} 
The derivation of (\ref{Okj}) is given in Appendix~\ref{ApB}.

\subsection{Problem Transformation}\label{S5.2}
	
\subsubsection{Transformation of $\mathbf{x}^{\prime}$ and $\mathbf{x}$}

	Consider the constant modulus constellation using phase shift keying (PSK) as an example. Since $x_{\operatorname{loc}_p+n}^{\prime}$ and $x_{\operatorname{loc}_p+n}$ take values from the $M$-PSK constellation, $r_{p,n}\! =\! r_{p,n}^{\prime}\! = 1$, and $\psi_{p,n}$ is given by 
\begin{equation}\label{eq-psi} % eq.55
	\psi_{p,n} = \frac{2\pi k}M,  k = -(M-1),\ldots,-1,0,1,\ldots,M-1,
\end{equation}
where each possible value is assigned with equal probability, i.e., $\psi_{p,n}$ follows a discrete uniform distribution. For the case where $\mathbf{x}^{\prime}\! \ne\! \mathbf{x}$, $O_{k,j}(\mathcal{P}_c)$ can be expressed as 
\begin{align}\label{eqTMx} % eq.56
	& O_{k,j}^{(\mathbf{x}^{\prime} \ne \mathbf{x})}(\mathcal{P}_c) =  \sum_{\mathbf{x}^{\prime} \ne \mathbf{x}}\sum_{r=1}^{\mathcal{R}}\sum_{p=1}^P\sum_{n=0}^{N-1}\left( - \cos\left(\psi_{p,n}+\theta_n^{\prime}-\theta_n \right) \right) \nonumber \\
	&\hspace*{3mm} = \sum_{r=1}^{\mathcal{R}} \sum_{p=1}^P\sum_{n=0}^{N-1}\!\! \bigg(\! \sum_{k=1-M}^{M-1}\!\! \cos\Big(\frac{2\pi k}M + \theta_n^{\prime}-\theta_n \Big)\! \bigg) = 0.
\end{align}
Therefore, for the objective function $O_{k,j}(\mathcal{P}_c)$, it is sufficient to consider the case of $\mathbf{x}^{\prime}\! =\! \mathbf{x}$. In this case, $\psi_{p,n}\! =\! 0$ for all $n\! =\! 0,1,\ldots,N-1$, and $O_{k,j}(\mathcal{P}_c)$ can be expressed as
\begin{equation}\label{O_PSK} % eq.57
	O_{k,j}(\mathcal{P}_c) = \sum_{r=1}^{\mathcal{R}}\sum_{p=1}^P\sum_{n=0}^{N-1}\left(1-\cos\left(\theta_n^{\prime}-\theta_n\right)\right). 
\end{equation}

For non-constant modulus constellations, take quadrature amplitude modulation (QAM) as an example. Then $x_{\operatorname{loc}_p+n}^{\prime}$ and $x_{\operatorname{loc}_p+n}$ take values from the $M$-QAM constellation and the constellation diagram is symmetric with respect to the origin. Hence, for each pair of $x_{\operatorname{loc}_p+n}^{\prime}$ and $x_{\operatorname{loc}_p+n}$  with phase difference $\psi_{p,n} = \gamma$, there exist a pair of corresponding constellation points with phase difference $\psi_{\hat{p},\hat{n}} = \gamma + \pi$. Therefore, for the case of $\mathbf{x}^{\prime} \ne \mathbf{x}$, we can obtain
\begin{align}\label{O_QAM_ne} % eq.58
	& \sum_{\mathbf{x}^{\prime} \ne \mathbf{x}}\sum_{r=1}^{\mathcal{R}}\sum_{p=1}^P\sum_{n=0}^{N-1}\left(r_{p,n}r_{p,n}^{\prime} \cos \left(\theta_n^{\prime} \!-\! \theta_n + \psi_{p,n} \right)\right) \nonumber \\
	& = \sum_{r=1}^{\mathcal{R}} \sum_{p=1}^P\sum_{n=0}^{N-1}\! \Big(\! \sum_{\psi_{p,n}}\!\! \cos\left(\gamma \!+\! \theta_n^{\prime} \!-\! \theta_n \right)\! + \!\cos\left(\gamma \!+\! \pi \!+\! \theta_n^{\prime}\!-\!\theta_n \right)\!\Big) \nonumber \\
	& = 0.
\end{align}
Consequently, it is sufficient to focus on the case of $\mathbf{x}^{\prime}\! =\! \mathbf{x}$ for the objective function $O_{k,j}(\mathcal{P}_c)$. Under this condition, $r_{p,n}\! =\! r_{p,n}^{\prime}$ and $\psi_{p,n}\! =\! 0$, and $O_{k,j}(\mathcal{P}_c)$ can be expressed as
\begin{align}\label{O_QAM} % eq.59
	& O_{k,j}(\mathcal{P}_c) \!=\! \sum_{\mathbf{x}^{\prime}, \mathbf{x}}\sum_{r=1}^{\mathcal{R}} \sum_{p=1}^P\sum_{n=0}^{N-1} 2r_{p,n}^2\left(1-\cos\left(\theta_n^{\prime}-\theta_n\right)\right),
\end{align}
where signal amplitudes $r_{p,n}$ are unrelated to optimization.

In summary, for both constant modulus and non-constant modulus constellations, the optimization problem (\ref{problem_origin}) can be transformed into
\begin{subequations}\label{problem_x} %eq.60a, eq.60b
	\begin{align}
		& \max_{\mathcal{P}_c} \min_{k,j \in [1,2^{b_2}]} \sum_{r=1}^{\mathcal{R}}\sum_{p=1}^P\sum_{n=0}^{N-1}\left(1-\cos\left(\theta_n^{\prime}-\theta_n\right)\right) , \label{problem_xa} \\
		& \text{ s.t. }\quad (\ref{constraint1}), (\ref{constraint2}), (\ref{constraint3}).
	\end{align}
\end{subequations}

\subsubsection{Transformation of $\mathrm{P}_{c_2}^{(j)}$ and $\mathrm{P}_{c_2}^{(k)}$} 
	
	In the problem (\ref{problem_x}), each pair of indices of PCPGs, i.e., each pair of $j,k$, corresponds to a specific set of $\mathrm{P}_{c_2}^{(j)}$ and $\mathrm{P}_{c_2}^{(k)}$. By substituting (\ref{theta_define}) into the objective function of (\ref{problem_x}), $O_{k,j}(\mathcal{P}_c)$ becomes 
\begin{equation}\label{O54} % eq.62
	O_{k,j}(\mathcal{P}_c) = \sum_{r=1}^{\mathcal{R}}\sum_{p=1}^P\sum_{n=0}^{N-1} \left(1-O^{\prime}\right),
\end{equation} 
where 
\begin{equation}\label{Delta56} % eq.62
	\left\{
		\begin{aligned}
			& O^{\prime} \!=\! \cos \! \left(2\pi  \left(\Delta c_{2,\operatorname{loc}_p+n} (\operatorname{loc}_p+n)^2 \!-\! \Delta c_{2,n}n^2\right)\right), \\
			& \Delta c_{2,\operatorname{loc}_p+n} = c_{2,\operatorname{loc}_p+n}^{'}-c_{2,\operatorname{loc}_p+n},\\
			& \Delta c_{2,n} = c_{2,n}^{'}-c_{2,n}.
		\end{aligned}
	\right.
\end{equation}   
It can be observed that in the case of $j \ne k$, i.e., $\mathrm{P}_{c_2}^{(j)} \ne \mathrm{P}_{c_2}^{(k)}$, as the discrepancy between $\mathrm{P}_{c_2}^{(j)}$ and $\mathrm{P}_{c_2}^{(k)}$ diminishes, the number of zero values in $O^{\prime}$ increases. Therefore, for the problem (\ref{problem_x}), it is sufficient to consider the case where $\big\|\mathrm{P}_{c_2}^{(j)} - \mathrm{P}_{c_2}^{(k)}\big\|_{0}=2$. Hence, the optimization problem can be further expressed as   
\begin{subequations}\label{problem_final} %eq.63a-eq.63.e
	\begin{align}
		\max_{\mathcal{P}_c} \min_{k,j} & \sum_{r=1}^{\mathcal{R}} \sum_{p=1}^P \sum_{n=0}^{N-1} \left(1 - \cos\left(\theta_n^{\prime} - \theta_n\right)\right), \label{problem_finala} \\
		\text{s.t.} \quad & j \ne k, \label{problem_finalb} \\
		& \big\|\mathrm{P}_{c_2}^{(j)} - \mathrm{P}_{c_2}^{(k)}\big\|_{0}=2, \label{problem_finalc} \\
		& c_2 \in (0,\, 1), \\
		& c_2 \in \mathbb{I}.
	\end{align}
\end{subequations}

\subsection{Problem Solver}\label{S5.3}

	The optimization (\ref{problem_final}) is a non-convex problem, and it is challenging to obtain a global optimal solution. To this end, the PSO-based algorithm is invoked to obtain a suboptimal solution, attributed to its rapid convergence and exemplary global searching capabilities.

\begin{algorithm}[t]
\caption{PSO-Based Algorithm for $c_2$ Alphabet Design}
\label{Alg1}
\begin{algorithmic}[1] %
	\Require $N_p$, $\varpi$, $\varrho_{global}$, $\varrho_{local}$, $v_{\max}$, $I_{\max}$, $\lambda$ and all other parameters required to evaluate fitness function;
	\Ensure $\mathcal{P}_c$;
	\State Set $i_{ter}\! =\! 0$ and initialize $N_p$ particles with positions $\mathcal{P}^{(0)}$ and zero velocities $\mathcal{V}^{(0)}$;
	\State Calculate fitness values of all particles by (\ref{utility-function}), $\mathcal{F}_P\Big(\bm{p}_{n_p}^{(0)}\Big)$, $n_p = 1, 2, \ldots, N_p$;
	\State Initialize local optimal position of each particle $\bm{p}_{n_p, local} = \bm{p}_{n_p}^{(0)}$, calculate global optimal position $\bm{p}_{global} = \arg \max\limits_{1\le n_p\le N_p}\big({\cal F}_p(\bm{p}_{n_p}^{(0)})\big)$;
	\While{$i_{ter} \le I_{\mathrm{max}}$}
		\For{$n_p = 1$ to $N_p$}
			\State Update velocity $\bm{v}_{n_p}^{(i_{ter})}$ according to (\ref{update-velocity});
			\For{$i = 1$ to $\lambda$}
			  \If{$\bm{v}_{n_p}^{(i_{ter})}[i] > {v}_{\mathrm{max}}$}
			    \State $\bm{v}_{n_p}^{(i_{ter})}[i] \gets {v}_{\max}$;
			  \ElsIf{$\bm{v}_{n_p}^{(i_{ter})}[i] < -{v}_{\max}$}
			    \State $\bm{v}_{n_p}^{(i_{ter})}[i] \gets -{v}_{\max}$
			  \EndIf
		  \EndFor
		  	\State Update position $\bm{p}_{n_p}^{(i_{ter})}$ based on (\ref{update position});
			\State Calculate fitness value $\mathcal{F}_P\Big(\bm{p}_{n_p}^{(i_{ter})}\Big)$ using (\ref{utility-function});
			\If{$\mathcal{F}_P\Big(\bm{p}_{n_p}^{(i_{ter})}\Big) > \mathcal{F}_P\Big(\bm{p}_{n_p, local}\Big)$}
			  \State $\bm{p}_{n_p, local} \gets \bm{p}_{n_p}^{(i_{ter})}$;
			\EndIf
			\If{$\mathcal{F}_P\Big(\bm{p}_{n_p}^{(i_{ter})}\Big) > \mathcal{F}_P\Big(\bm{p}_{global}\Big)$}
			  \State $\bm{p}_{global} \gets \bm{p}_{n_p}^{(i_{ter})}$;
			\EndIf
		\EndFor
		\State $i_{ter} \gets i_{ter} + 1$;
	\EndWhile
	\State Obtain global optimal position $\mathcal{P}_c \gets \bm{p}_{global}$;
	\State \textbf{return} $\mathcal{P}_c$.
\end{algorithmic}
\end{algorithm}
	
	First, a population of $N_p$ particles with velocities and positions are initialized. The velocities of the particles are represented by $\mathcal{V}^{(0)}\! =\! \{\bm{v}_1^{(0)},\bm{v}_2^{(0)}\ldots,\bm{v}_{N_p}^{(0)}\}$, which is indicative of the extent of change occurring during the iterative process. The positions of the particles are denoted by $\mathcal{P}^{(0)}\! =\! \{\bm{p}_1^{(0)},\bm{p}_2^{(0)},\ldots,\bm{p}_{N_p}^{(0)}\}$, where each position represents a potential solution for the pre-chirp alphabet, i.e.,
\begin{equation}\label{position} % eq.64
	\bm{p}_{n_p}^{(0)} = \mathcal{P}_{c,n_p}^{(0)} = \left\{{c_{2,{n_p}}^{(0,1)}, c_{2,{n_p}}^{(0,2)}, \ldots, c_{2,{n_p}}^{(0,\lambda )}} \right\},
\end{equation}
in which $n_p$ denotes the index of particle. As an initial solution, the first particle $\bm{p}_{1}^{(0)}$ is initialized as a heuristic pre-chirp alphabet, where the elements in $\mathcal{P}_c$ are evenly distributed within the interval [0,\, 1], and the remaining particles are randomly initialized.

	Subsequently, the fitness value of each particle is evaluated per the specified utility function. In light of the constraints imposed by (\ref{problem_finalb}) and (\ref{problem_finalc}), a brick wall penalty factor is introduced, and the utility function in the $i_{ter} $-th iteration is defined as
\begin{equation}\label{utility-function} % eq.65
	\mathcal{F}_P\Big(\bm{p}_{n_p}^{(i_{ter})}\Big) = \begin{cases}\epsilon, &\text{if}\hspace{0.2cm} \bm{p}_{n_p}^{(i_{ter})} \hspace{0.1cm} \text{is feasible,}\\-1,&\text{otherwise,}\end{cases}
\end{equation}        
where $\epsilon\! =\! \min_{k,j} O_{k,j}(\mathcal{P}_c)$ and $O_{k,j}(\mathcal{P}_c)$ is calculated by (\ref{O54}). Then the particle with the greatest fitness value is considered to be the initial global optimal position $\bm{p}_{global}$, and the local optimal position $\bm{p}_{n_p, local}$ of each particle is initialized as $\bm{p}_{n_p}^{(0)}$.
	
	Each particle conveys its local optimal position to other particles during the iteration. The velocity and position of each particle are updated according to
\begin{align} % eqs.66,67
	\bm{v}_{n_p}^{(i_{ter})} =& \varpi \bm{v}_{n_p}^{(i_{ter}-1)} + r_1 \varrho_{local} \big(\bm{p}_{{n_p},local} -\bm{p}_{n_p}^{(i_{ter}-1)}\big) \nonumber \\
	& +r_2 \varrho_{global}\big(\bm{p}_{global} - \bm{p}_{n_p}^{(i_{ter}-1)}\big), \label{update-velocity} \\
	\bm{p}_{n_p}^{(i_{ter})} =& \bm{p}_{n_p}^{(i_{ter}-1)} + \bm{v}_{n_p}^{(i_{ter})} , \label{update position}
\end{align}                                                               
respectively, where $\varpi$ is the inertia weight, $\varrho_{global}$ and $\varrho_{local}$ denote the global and local updating coefficients, respectively, while $r_1$ and $r_2$ are two random values drawn within the interval [0,\, 1]. The particle velocity is constrained to the range between $-{v}_{\max}$ and ${v}_{\max}$. Afterward, the fitness values of all the particles are evaluated, and the global and local optimal positions, $\bm{p}_{global}$ and $\big\{\bm{p}_{n_p, local}\big\}_{n_p=1}^{N_p}$, are updated. The above update process is repeated until the maximum number $I_{\mathrm{max}}$ of iterations is reached. The particle with the highest fitness value corresponds to the optimization result $\mathcal{P}_c$. The procedure of the proposed PSO-based algorithm is outlined in Algorithm~\ref{Alg1}.
	
	%%%%%%%%%%%%%%%%%%%%%%%%%%%%%%%%%%%%%%%%%%%%%%%%%%%%%%%%%%%%%%%%%%
	%\newpage
\section{Simulation Results and Analysis}\label{S6}
	
	In this section, we carry out simulations to evaluate the performance of the proposed AFDM-PIM schemes, and conduct the BER performance comparison between the proposed AFDM-PIM scheme and other existing counterparts. The accuracy of the PEP-based theoretical derivation is also investigated in comparison with the Monte Carlo simulation results.
  In simulations, the carrier frequency is set to $f_{c}\! =\! 8$\,GHz, and the subcarrier spacing in the DAF domain is set to $f_{s}\! =\! 1.5$\,kHz. Two distinct simulation scenarios are conducted, namely, when the full diversity condition (\ref{condition1}) is satisfied and when it is not satisfied. The maximum normalized Doppler shifts for these two scenarios are set to $\alpha_{\mathrm{max}}\! =\! 1$ and $\alpha_{\mathrm{max}}\! =\! 2$, corresponding to the high-speed scenarios with the maximum speed of mobile station $v_e\! =\! 202.5$\,km/h and the ultra high-speed scenarios with the maximum speed of mobile station $v_e\! =\! 405$\,km/h, respectively.\footnote{The relationship between speed and normalized Doppler shift is shown in Appendix~\ref{calc_doppler}.} The Doppler shift of each path is generated according to Jakes Doppler spectrum approach as $\alpha_{p}\! =\! \alpha_{\mathrm{max}} \cos(\theta_{p,d})$, where $\theta_{p,d}\! \in\! [-\pi,\,\pi]$ (for integer Doppler cases, the Doppler shift is $\lfloor \alpha_{\mathrm{max}} \cos(\theta_{p,d}) \rfloor$). We set the pertinent parameters of Algorithm~\ref{Alg1} as ${v}_{\max}\! =\! 0.05$, $\varpi\! =\! 0.5$, $N_p\! =\! 200$, $I_{\mathrm{max}}\! =\! 300$, $\varrho_{global}\! =\! 2$ and $\varrho_{local}\! =\! 2$.  The adopted simulation parameters are listed in Table~\ref{table:simulation_parameters}. Unless otherwise specified, the ML detector is employed for both the proposed AFDM-PIM and the classical benchmarks. 
	
\begin{table}[!t]
	\centering
	\caption{Simulation Parameters}
	\label{table:simulation_parameters} % Tabl.III
	\vspace*{-2mm}
	\resizebox{\linewidth}{!}{
	\begin{tabular}{|c|c|c|}
		\hline
		\textbf{Parameter} & \textbf{Description} & \textbf{Value} \\ \hline
		$c$ & Speed of light in free space & $3 \times 10^8$ m/s \\ \hline
		$f_c$ & Carrier frequency & 8 GHz \\ \hline
		$f_s$ & Subcarrier spacing & 1.5 kHz \\ \hline
		$v_e$ & Maximum speed of mobile station & 202.5/405\,km/h \\ \hline
		${v}_{\max}$ & Maximum particle velocity & 0.05 \\ \hline
		$\mathcal{P}_c$ & \makecell {Alphabet of $c_2$ realizations \\ for $\lambda$ = 2/3/4} & \makecell{\{0.20, 0.60\} \\  \{0.29, 0.62, 0.93\} \\ \{0.01, 0.20, 0.41, 0.80\} } \\ \hline
		$\varpi$ & Inertia weight & 0.5 \\ \hline
		$N_p$ & Number of particles & 200 \\ \hline
		$I_{\max}$ & Maximum iteration value & 300 \\ \hline
		$\varrho_{global}$ & Global updating coefficients & 2 \\ \hline
		$\varrho_{local}$ & Local updating coefficients & 2 \\ \hline
	\end{tabular}
	}
	\vspace*{-3mm}
\end{table}

\begin{figure}[!h] 
\vspace*{-4mm}
\centering
\includegraphics[width=3.7in]{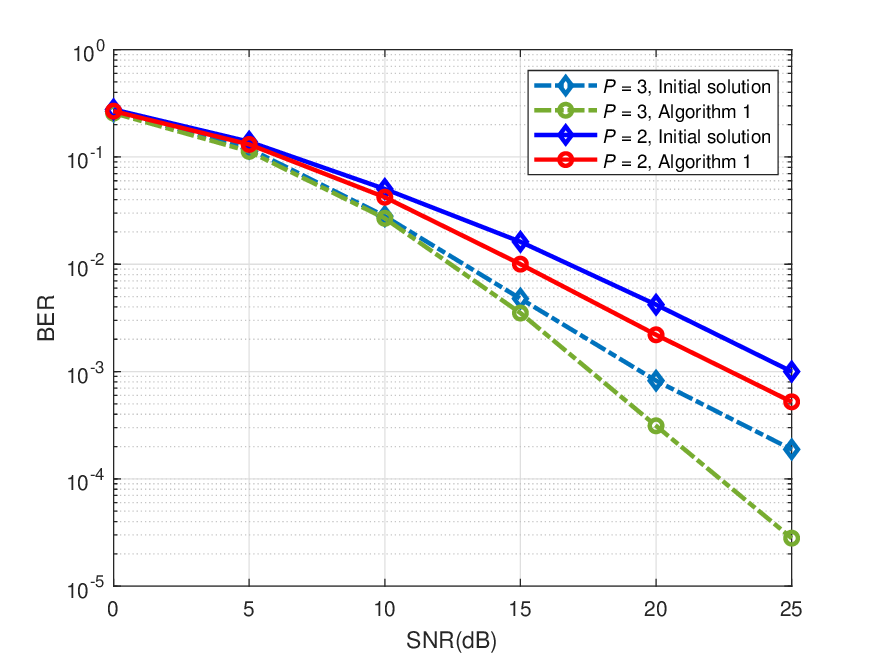}
\vspace*{-7mm}
\caption{BER Performance comparison between the $c_2$ alphabet designed by Algorithm~\ref{Alg1} and the initial $c_2$ solution, given two different numbers of paths.}
\label{vsPSO} % Fig.4
\vspace*{-4mm}
\end{figure}

\subsection{Effectiveness of Algorithm~\ref{Alg1}}\label{S6.1}
	
  We first validate the effectiveness of the proposed $c_2$ alphabet design with Algorithm~\ref{Alg1}. Fig.~\ref{vsPSO} compares the BER performance achieved by the $c_2$ alphabet designed by Algorithm~\ref{Alg1} and that attained by the initial solution, namely, the first initial particle $\bm{p}_{1}^{(0)}$ for Algorithm~\ref{Alg1}, given two numbers of paths $P\! =\! 2$ and 3. Binary PSK (BPSK) is employed as the modulation scheme for the data bits, and the parameter settings are $(N, G, \lambda, d_{\mathrm{max}},\alpha_{\mathrm{max}})\! =\! (6, 2, 3, 1, 1)$. Since the full diversity order condition is satisfied, the BER performance improves with the increase in the number of paths. 
	Fig.~\ref{vsPSO} indicates that the $c_2$ alphabet designed by Algorithm~\ref{Alg1} exhibits superior performance by about 3\,dB in SNR than the initial solution at the BER level of $10^{-3}$, given $P\! =\!3$. The results of Fig.~\ref{vsPSO} demonstrate the effectiveness of Algorithm~\ref{Alg1}.

\begin{figure}[!t]
	\vspace*{-2mm}
	\centering
	\includegraphics[width=3.7in]{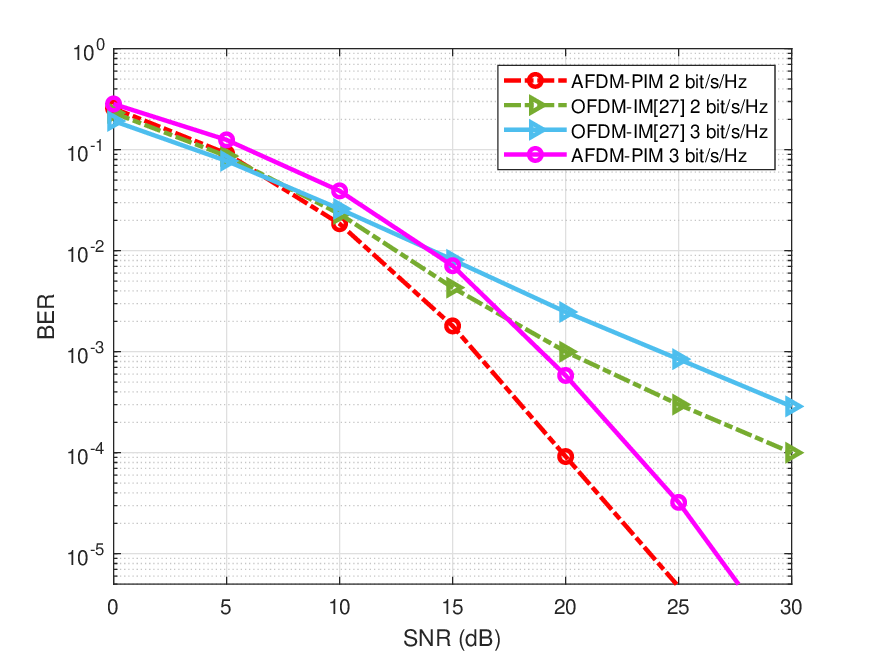}
	\vspace*{-7mm}
	\caption{Performance comparison between the proposed AFDM-PIM scheme and the OFDM-IM scheme given two different spectral efficiencies.}
	\label{vsOFDMIM} % Fig.5
	\vspace*{-3mm}
\end{figure}

\begin{figure}[!b]
\vspace*{-6mm}
\centering
\includegraphics[width=3.7in]{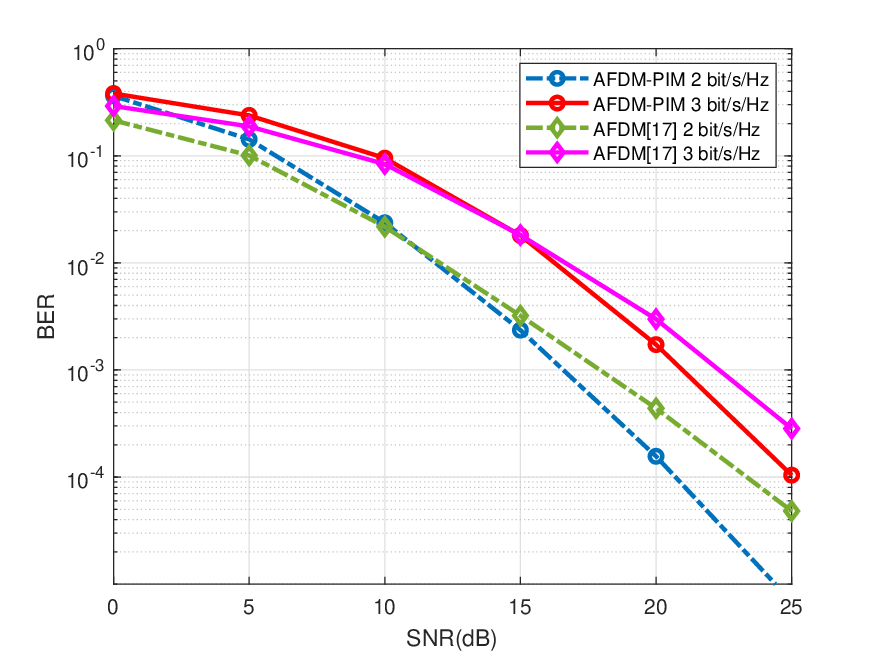}
\vspace*{-7mm}
\caption{Performance comparison between the proposed AFDM-PIM scheme and the classical AFDM scheme given two different spectral efficiencies.}
\label{vsAFDM} % Fig.6
\vspace*{-1mm}
\end{figure}

\subsection{Performance Comparison with Existing Benchmarks}\label{S6.2}

	Fig.~\ref{vsOFDMIM} compares the BER performance achieved by the proposed AFDM-PIM and the OFDM-IM \cite{bacsar2013orthogonal} under a doubly dispersive channel with $P\! =\! 4$ and given two SEs of $2$ and $3$\,bit/s/Hz. For the proposed AFDM-PIM, $(N, G, \lambda, d_{\mathrm{max}}, \alpha_{\mathrm{max}})$ are set to $ (8, 2, 4, 0, 1)$, and BPSK and quadrature PSK (QPSK) are employed to reach the SEs of $2$ and $3$ bit/s/Hz, respectively. For the OFDM-IM scheme \cite{bacsar2013orthogonal}, each group comprises $n$ subcarriers, with $a$ subcarriers activated at each transmission. To reach the SEs of $2$ and $3$ bit/s/Hz, $(n,a)\! =\! (4,2)$ and $(n,a)\! =\! (8,7)$ are employed, respectively, and $8$-PSK is utilized. It can be seen from Fig.~\ref{vsOFDMIM} that at the BER level of $10^{-3}$, our AFDM-PIM attain the SNR gains of about 5\,dB and 4\,dB over the OFDM-IM for the cases of $\text{SE}\! =\! 3$\,bits/s/Hz and 2\,bits/s/Hz, respectively. This enhanced performance is attributable to the intrinsic advantage of AFDM-PIM, which can achieve full diversity order on doubly dispersive channels.

	Fig.~\ref{vsAFDM} compares the BER performance of our AFDM-PIM with that of the classical AFDM scheme \cite{bemani2023affine} under the same doubly dispersive channel with $P\! =\! 4$. The proposed AFDM-PIM employs BPSK and QPSK to achieve the SEs of $2$ and $3$\,bit/s/Hz, respectively. The other parameters are as $(N, G, \lambda, d_{\mathrm{max}}, \alpha_{\mathrm{max}})\!  =\! (8, 2, 4, 2, 2)$. To reach the same SE levels in the classical AFDM \cite{bemani2023affine}, the number of subcarriers is set to $8$, and QPSK and $8$-PSK are utilized correspondingly. 
	The both schemes do not satisfy the full diversity order condition. Specifically, the proposed AFDM-PIM does not satisfy Condition 1 in Theorem~\ref{T3}, while the AFDM does not satisfy Theorem~1 of \cite{bemani2023affine}.
	It can be seen from Fig.~\ref{vsAFDM} that the AFDM-PIM demonstrates an about $2$\,dB gain in the SNR compared to the AFDM scheme at the BER level of $10^{-3}$. This is because the AFDM-PIM allows lower order constellations to reach the same SE as the AFDM. The results of Fig.~\ref{vsAFDM} indicates that the AFDM-PIM offers a viable alternative for communication under doubly dispersive channels.
	
\begin{figure}[!h] 
\vspace*{-4mm}
\centering
\includegraphics[width=3.7in]{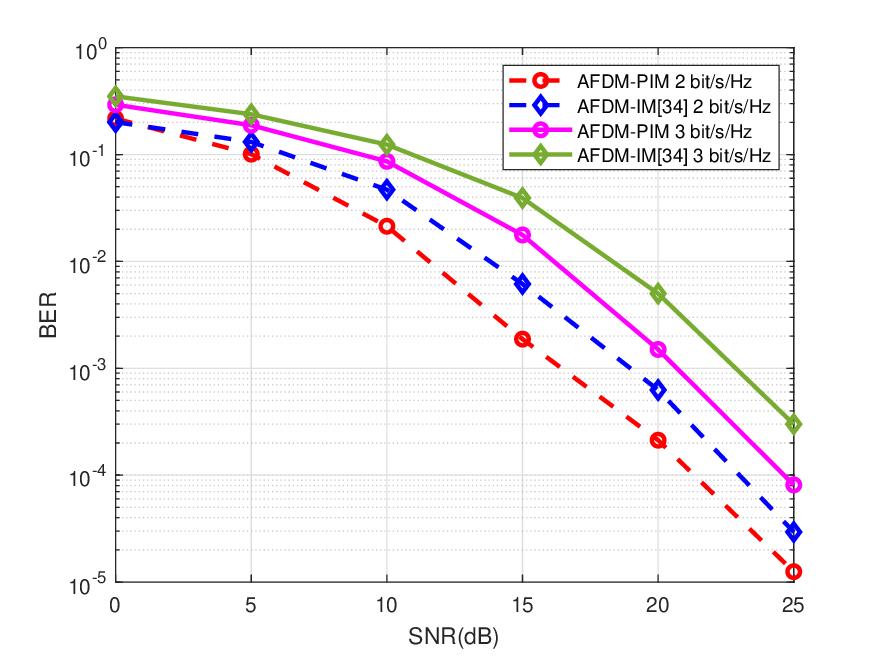}
\vspace*{-7mm}
\caption{Performance comparison between the proposed AFDM-PIM scheme and the AFDM-IM scheme given two different spectral efficiencies.}
\label{vsAFDMIM} % Fig.7
%\vspace*{-1mm}	
\end{figure}
	
  In Fig.~\ref{vsAFDMIM}, we evaluate the BER performance of the proposed AFDM-PIM and the AFDM-IM \cite{10342712} under a doubly dispersive channel with the number of paths $P\! =\! 3$ and given two SEs of $2$ and $3$ bit/s/Hz.
	For the proposed AFDM-PIM, $(N, G, \lambda, d_{\mathrm{max}}, \alpha_{\mathrm{max}})$ are set to $ (8, 2, 4, 1, 2)$, and BPSK and QPSK are employed to reach the SEs of $2$ and $3$ bit/s/Hz, respectively. For the AFDM-IM scheme \cite{10342712}, each group comprises $n$ subcarriers, with $a$ subcarriers activated at each transmission. To reach the same SE levels for the AFDM-IM, $(n,a)\! =\! (4,2)$ and $(n,a)\! =\! (8,7)$ are employed, respectively, and $8$-PSK is utilized. Note that the both schemes do not meet the full diversity order condition.
	It can be seen from Fig.~\ref{vsAFDMIM} that our AFDM-PIM scheme yields better BER performance than the AFDM-IM scheme, with an SNR gain of more than $2$\,dB at the BER level of $10^{-3}$. The enhanced performance is attributable to the capability of our AFDM-PIM to employ lower order constellations than the AFDM-PIM, while  reaching the same SE as the AFDM-IM.
	
\begin{figure}[!t]
\vspace*{-3mm}
\centering
\includegraphics[width=3.7in]{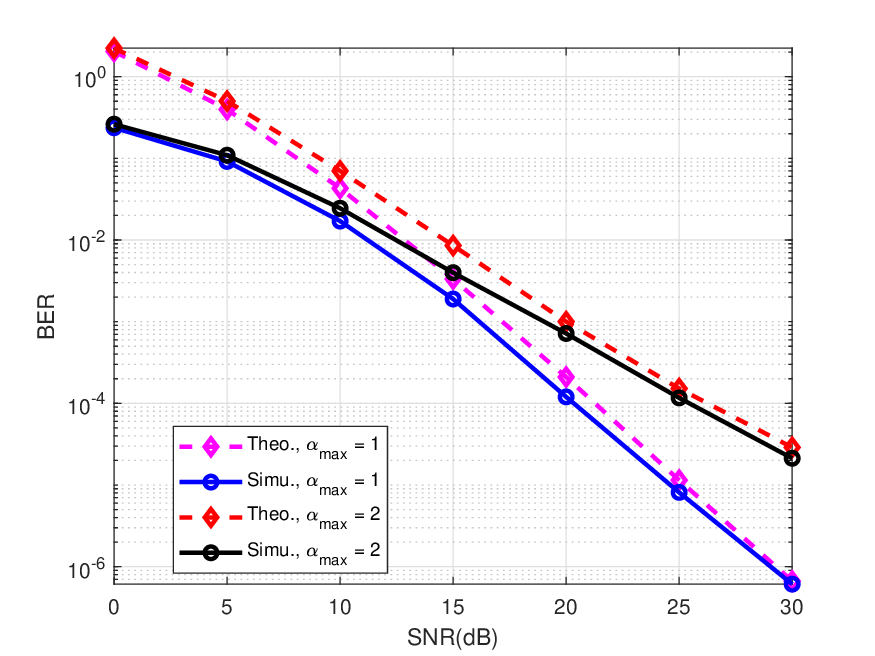}
\vspace*{-7mm}
\caption{Comparison of the theoretical ABEP upper bound and the simulated BER of the proposed AFDM-PIM given two different Doppler shifts.}
\label{vsTheo.}	% Fig.8
\vspace*{-3mm}
\end{figure}
		
\begin{figure*}[!tbp]\setcounter{equation}{69}
\vspace*{-1mm}
\begin{equation}\label{phi_a} % eq.70
	\begin{bmatrix}\hat H_a[0,\mathrm{loc}_a]\hat{x}_{\mathrm{loc}_a}-H_a[0,\mathrm{loc}_a]x_{\mathrm{loc}_a} \\ \hat H_a[1,(\mathrm{loc}_a+1)_N]\hat{x}_{(\mathrm{loc}_a +1)_N } - H_a[1,(\mathrm{loc}_a+1)_N]x_{(\mathrm{loc}_a + 1)_N} \\ \vdots \\ \hat H_a[N - 1,(\mathrm{loc}_a+N-1)_N]\hat{x}_{(\mathrm{loc}_a+N-1)_N}-H_a[N - 1,(\mathrm{loc}_a+N-1)_N]x_{(\mathrm{loc}_a+N-1)_N}\end{bmatrix},
\end{equation}
\begin{equation}\label{phi_b} % eq.71
	\begin{bmatrix}\hat H_b[0,\mathrm{loc}_b]\hat{x}_{\mathrm{loc}_b} - H_b[0,\mathrm{loc}_b]x_{\mathrm{loc}_b} \\ \hat H_b[1,(\mathrm{loc}_b+1)_N]\hat{x}_{(\mathrm{loc}_b +1)_N} -H_b[1,(\mathrm{loc}_b+1)_N]x_{(\mathrm{loc}_b + 1)_N} \\ \vdots \\ \hat H_b[N - 1,(\mathrm{loc}_b+N-1)_N]\hat{x}_{(\mathrm{loc}_b+N-1)_N}-H_b[N - 1,(\mathrm{loc}_b+N-1)_N]x_{(\mathrm{loc}_b+N-1)_N}\end{bmatrix},
\end{equation}
\hrulefill
\vspace*{-4mm}
\end{figure*}

\subsection{Accuracy of PEP-based Theoretical Analysis}\label{S6.3}
 	
	Fig.~\ref{vsTheo.} presents a comparison of the theoretical ABEP upper bound and the simulated BER of the proposed AFDM-PIM scheme over a doubly dispersive channel with the number of paths $P\! \!= 3$. The spectral efficiency of 1.5 bit/s/Hz is considered, $(N, G, \lambda, d_{\mathrm{max}})$ are set to $(4, 2, 2, 0)$ and BPSK is employed as the modulation scheme. Besides, the maximum normalized Doppler shifts are set to $\alpha_{\mathrm{max}}\! =\! 1$ and $\alpha_{\mathrm{max}}\! =\! 2$, corresponding to the high-speed and ultra high-speed scenarios, respectively. It can be seen from Fig.~\ref{vsTheo.} that in the both cases, the theoretical ABEP results deviate from the simulated results in the low-SNR region. This is because the theoretical ABEP upper bound is subject to several approximations, which inevitably becomes inaccurate when the noise is dominant. On the other hand, the theoretical results closely match the simulated BER curves at sufficiently high SNRs, which demonstrates the validity of the PEP analysis for the AFDM-PIM presented in Subsection~\ref{S4.1}.
	
	%%%%%%%%%%%%%%%%%%%%%%%%%%%%%%%%%%%%%%%%%%%%%%%%%%%%%%%%%%%%%%%%%%
	% \newpage
\section{Conclusions}\label{S7}

	In this paper, we have proposed the novel AFDM-PIM scheme to combat time-frequency doubly selective channel fading. In our proposed scheme, the pre-chirp parameters on subcarriers are no longer fixed but are selected from a predefined alphabet. This enables additional bits to be transmitted through indexing the specific pre-chirp parameter values on subcarriers, thereby enhancing both spectrum and energy efficiency. The PEP-based theoretical BER upper bound for the proposed AFDM-PIM scheme with the ML detection has been derived and verified through simulations, and the full diversity order conditions of the proposed AFDM-PIM under doubly dispersive channels have been derived. Furthermore, we have presented a PSO-based optimization algorithm to design the pre-chirp parameter alphabet. Both analytical and simulation results have demonstrated that the proposed AFDM-PIM exhibits enhanced spectral efficiency and superior error performance in comparison to classical multi-carrier modulation schemes. 

\appendix

\subsection{Proof of Theorem~\ref{T3}}\label{ApA}	
 
\begin{proof}
	First, we show that Condition 1 is necessary for the AFDM-PIM to achieve the full diversity order. Assume that  Condition 1 is not met, i.e.,\setcounter{equation}{67}
\begin{align}\label{eqA1} % eq.68
  & P > (d_{\mathrm{max}}+1)(2\alpha_{\mathrm{max}}+1) .
\end{align}
Under this assumption, according to the definition of $\mathrm{loc}_p\! =\! (\alpha_{p} + 2 N c_1 d_p)_{N}$, the following situation must be true
\begin{equation}\label{eqA2} % eq.69
	\exists a,b\in[1,\ldots,P], a \ne b, ~\text{such that} ~\mathrm{loc}_a = \mathrm{loc}_b.
\end{equation}
The corresponding two columns in the matrix $\mathbf{\Phi}(\delta)$ can be expressed as (\ref{phi_a}) and (\ref{phi_b}) at the top of this page, where $\hat{H}_{(\cdot)}[\cdot]\hat{x}_{(\cdot)}$ and $H_{(\cdot)}[\cdot]x_{(\cdot)}$ represent the corresponding elements in $\hat{\mathbf{\Phi}} (\hat{\mathbf{x}})$ and $\mathbf{\Phi}(\mathbf{x})$, respectively.

\begin{figure*}[!bp]\setcounter{equation}{79}
\vspace*{-4mm}
\hrulefill
\begin{align}\label{F-norm}  % eq.80
	& \left\|(\mathbf{\Phi}_k^r (\mathbf{x}^{\prime})-\mathbf{\Phi}_j^r(\mathbf{x}))\right\|_{\mathrm{F}}^2 =\sum_{p=1}^P \sum_{n=1}^{N} \left| H_p^{\prime}\left[n,\operatorname{loc}_p+n\right] x^{\prime}_{\operatorname{loc}_p+n} - H_p\left[n,\operatorname{loc}_p+n\right] x_{\operatorname{loc}_p+n} \right|^2 \nonumber \\
	&\hspace*{10mm} = \sum_{p=1}^P \sum_{n=1}^{N} \left|e^{\textsf{j} \frac{2 \pi}{N}\left[N c_{2,\operatorname{loc}_p+n}^{\prime}\left(\operatorname{loc}_p+n\right)^2-{N} c_{2,n}^{\prime} n^2\right]} x^{\prime}_{\operatorname{loc}_p+n}] - e^{\textsf{j} \frac{2 \pi}{{N}}\left[{N} c_{2,\operatorname{loc}_p+n}\left(\operatorname{loc}_p+n\right)^2-{N} c_{2,n} n^2\right]} x_{\operatorname{loc}_p+n}]\right|^2\!\! ,	\!
\end{align} 
\vspace*{-1mm}
\end{figure*}  

	Based on (\ref{H_p}), the positions of the non-zero entries in the matrices $\mathbf{H}_a,\hat{\mathbf{H}}_a,\mathbf{H}_b$ and $\hat{\mathbf{H}}_b$ are consistent. In the instances, such as when both $\hat{\mathbf{\Phi}} (\hat{\mathbf{x}})$ and $\mathbf{\Phi}(\mathbf{x})$ contain only a single non-zero element at the same position, (\ref{phi_a}) and (\ref{phi_b}) are linearly correlated. Therefore, $\mathbf{\Phi}(\delta)$ cannot be full rank, i.e., the assumption (\ref{eqA1}) is false.
		
	Besides, when $\mathbf{\Phi}(\delta)$ achieves the full diversity order, the channel paths with different delay values or distinct Doppler frequency shifts are distinguished within the DAFT domain, as illustrated in Fig.~\ref{fig3}, i.e., $ (d_{\mathrm{max}}+1)(2\alpha_{\mathrm{max}}+1) \le N$. Therefore, Condition 1 is a prerequisite for achieving the full diversity order.
		
	Next we prove that Condition 2 can ensure that $\mathbf{\Phi}(\delta)$ has the full rank. $\mathbf{\Phi}(\delta)$ can be expressed as $\mathbf{\Phi}(\delta)\! =\! \hat{\mathbf{\Phi}}(\hat{\mathbf{x}})\! -\! \mathbf{\Phi}(\mathbf{x})\! =\! [\bm{\gamma}_1, \ldots, \bm{\gamma}_p, \ldots, \bm{\gamma}_{P}]$,  where $\bm{\gamma}_p\! \in\! \mathbb{C}^{N\times 1}$ with the entries $\bm{\gamma}_p[n]\! =\! \hat H_p\left[n,\left(\operatorname{loc}_p + n\right)_{N}\right] \hat{x}_{\left(\operatorname{loc}_p+n\right)_{N}}\! -\! H_p\left[n,\left(\operatorname{loc}_p+n\right)_{N}\right] x_{\left(\operatorname{loc}_p+n\right)_{N}}$, $p\! =\! 1,2,\ldots,P$, $n\! =\! 0, 1, \ldots, N-1$.\footnote{For the convenience of typesetting, we use $\mathrm{loc}_p+n$ to represent $\left(\operatorname{loc}_p+n\right)_{N}$ throughout the following text.}
		
	Given constants $\ell_{n,m}$, $\hat H_p\left[n,\operatorname{loc}_p + n\right]$ can be expressed as $\hat H_p[n,\operatorname{loc}_p+n]=\ell_{n,m}H_p[n,\operatorname{loc}_p+n]$. According to (\ref{H_p}), the entries of $\boldsymbol{\gamma}_p$ are given by\setcounter{equation}{71}
\begin{align}\label{eqAp5} % eq.72
	\bm{\gamma}_p[n] =& H_p[n,\mathrm{loc}_p+n] {\rho}_{n,m} \nonumber \\
	=& e^{\textsf{j} 2 \pi \left[ c_{2,\operatorname{loc}_p+n}\left(\operatorname{loc}_p+n\right)^2 - c_{2, n}n^2 \right]} \nonumber \\
	& \times e^{\textsf{j} \frac{2 \pi}{N} [-( \operatorname{loc}_p+n) + c_1 d_p^2 ]} {\rho}_{n,m},
\end{align}       
where ${\rho}_{n,m} = \left(x_{\mathrm{loc}_p+n}-\ell_{m,n}\hat{x}_{\mathrm{loc}_p+n}\right)$. Now assume that a set of numbers $\{\beta_p\}_{p=1}^P$ that are not all zero satisfy
\begin{equation}\label{correlated} % eq.73
	\sum_{p=1}^{P}\beta_p \bm{\gamma}_p=\mathbf{0}.
\end{equation} 
Without loss of generality, $\beta_b$ is assumed to be the non-zero number, $b\! \in\! \{1, 2, \ldots, P\}$. The relationship for any index $n$ can be obtained as
\begin{equation}\label{sum_n} % eq.74
	\sum_{p=1}^{P}\beta_p  \bm{\gamma}_p[n] = \sum_{p=1}^{P}\frac{\beta_p}{\beta_b}  \bm{\gamma}_p[n] =0,\quad\forall n. 
\end{equation}			
From (\ref{sum_n}), it is not difficult to obtain
\begin{align}\label{rho_2} % eq.75
	& {\rho}_{n,b} = - \!\!\! \sum_{p=1, p \neq b}^{P} \frac{\beta_p}{\beta_b} \frac{\boldsymbol{\gamma}_p[n]}{H_p[n,\mathrm{loc}_p+n]} \nonumber \\
	& = -e^{-\textsf{j} 2\pi c_{2,\mathrm{loc}_b}\mathrm{loc}_b^2} e^{\textsf{j} \frac{2\pi}{N}(N c_1(-d_b^2+\mathrm{loc}_b d_b))} \nonumber \\
	& \hspace{2mm} \times \!\!\! \sum_{p=1, p \neq b}^P \!\!\! e^{\textsf{j} 2\pi\big( c_{2,\mathrm{loc}_p+n}{(\mathrm{loc}_p+n)}^2 + c_1 d_p^2-\frac{\mathrm{loc}_p d_p}{N}\big)} \frac{\beta_p}{\beta_b} {\rho}_{n,m}.
\end{align}
Since ${\rho}_{0,b}$ is a signal error term unrelated to $c_{2,n}$, the phase of $\frac{\beta_p}{\beta_b}(m \neq b)$ must contain $2\pi\big(c_{2,\mathrm{loc}_b}\mathrm{loc}_b^2-c_{2,\mathrm{loc}_p+n}(\mathrm{loc}_p+n)^2\big)$ to eliminate the influence of irrational numbers. On the other hand, if (\ref{rho_2}) holds, there exists another non-zero $\beta_a$ ($a \ne b$). Similar to the derivation process of (\ref{rho_2}), the phase of $\frac{\beta_p}{\beta_a}(m \neq a)$ must contain $2\pi\big(c_{2,\mathrm{loc}_a}\mathrm{loc}_a^2-c_{2,\mathrm{loc}_p+n}(\mathrm{loc}_p+n)^2\big)$. Hence, the relationship between $\frac{\beta_a}{\beta_b}$ and $\frac{\beta_b}{\beta_a}$ can be obtained as $\frac{\beta_a}{\beta_b} \cdot \frac{\beta_b}{\beta_a} = 1$, i.e.,
\begin{align}\label{eqAf} % eq.76
	& e^{\textsf{j}2\pi c_{2,\mathrm{loc}_b}\mathrm{loc}_b^2} e^{-\textsf{j} 2\pi c_{2,\mathrm{loc}_a + n}(\mathrm{loc}_a + n)^2} \nonumber \\
  & \hspace{4mm} \times e^{\textsf{j}2\pi c_{2,\mathrm{loc}_a}\mathrm{loc}_a^2} e^{-\textsf{j} 2\pi c_{2,\mathrm{loc}_a + n}(\mathrm{loc}_a + n)^2} \vartheta = 1, ~ \forall n,
\end{align}
where $\vartheta$ is a complex number whose phase does not contain $c_{2,n}$. Given that all the $c_{2,n}$ are irrational numbers, the phase on the left-hand side of (\ref{eqAf}) cannot be an integer multiplying $2 \pi$.
  This implies that the imaginary part is not zero, and (\ref{eqAf}) does not hold. Consequently, the assumption of (\ref{correlated}) is invalid, which means that the column vectors of the matrix $\mathbf{\Phi}(\delta)$ are linearly independent, i.e., the rank of $\mathbf{\Phi}(\delta)$ is $P$. Combining this result with the definition of diversity order (\ref{order}), it is concluded that the AFDM-PIM can achieve the full diversity order.  
\end{proof}

\subsection{Derivation of $O_{k,j}(\mathcal{P}_c)$ (\ref{Okj})}\label{ApB}

	Substituting $\mathbf{\Phi}_j^r(\mathbf{x})\! =\! [\mathbf{H}_{1}\mathbf{x}, \ldots, \mathbf{H}_{P}\mathbf{x}]$ and $\mathbf{\Phi}_k^r(\mathbf{x}^{\prime})\! =\! [\mathbf{H}_{1}^{\prime}\mathbf{x}^{\prime}, \ldots, \mathbf{H}_{P}^{\prime}\mathbf{x}^{\prime}]$ into (\ref{Okj_o}) leads to
\begin{equation}\label{problem3} % eq.77
	O_{k,j}(\mathcal{P}_c) = \sum_{\mathbf{x}^{\prime}\!, \mathbf{x}} \sum_{r=1}^{\mathcal{R}} \sum_{p=1}^P\left\|\mathbf{H}_p^{\prime}\mathbf{x}^{\prime}-\mathbf{H}_p\mathbf{x}\right\|_\mathrm{2} .
\end{equation}			
According to (\ref{H_p}), the elements of $\mathbf{H}_p\mathbf{x}$ can be expressed as       
\begin{align}\label{MatrixHx64}  % eq.78
	\mathbf{H}_p \mathbf{x}[n] =& H_p\left[n,\operatorname{loc}_p + n \right] x_{\mathrm{loc}_p + n} \nonumber \\
	=& e^{\textsf{j} 2 \pi \big( c_{2,\operatorname{loc}_p+n} (\operatorname{loc}_p+n)^2 - c_{2, n}n^2 d_p \big)} 	\nonumber \\	
	& \times e^{\textsf{j} \frac{2 \pi}{N} \big(-( \operatorname{loc}_p+n) + c_1 d_p^2\big)} x_{\mathrm{loc}_p + n} . 
\end{align}
Moreover, the elements of $\mathbf{H}_p^{\prime}\mathbf{x}^{\prime}-\mathbf{H}_p\mathbf{x}$ can be expressed as
\begin{align}\label{MatrixHx-}  % eq.79
	& \left(\mathbf{H}_p^{\prime}\mathbf{x}^{\prime} - \mathbf{H}_p\mathbf{x}\right)[n] = H_p^{\prime}\left[n,\operatorname{loc}_p+n\right] x^{\prime}_{\operatorname{loc}_p+n} \nonumber \\
	&\hspace*{8mm} - H_p\left[n,\operatorname{loc}_p+n\right] x_{\operatorname{loc}_p+n} \nonumber \\
	&\hspace*{5mm} = \bigg(e^{\textsf{j} \frac{2 \pi}{N}\big(N c_{2,\operatorname{loc}_p+n}^{\prime} (\operatorname{loc}_p+n)^2-{N} c_{2,n}^{\prime} n^2\big)} x^{\prime}_{\operatorname{loc}_p+n} \nonumber \\
	&\hspace*{8mm} -e^{\textsf{j} \frac{2 \pi}{{N}}\big({N} c_{2,\operatorname{loc}_p+n} (\operatorname{loc}_p+n)^2-{N} c_{2,n} n^2\big)} x_{\operatorname{loc}_p+n}\bigg) \nonumber \\
	&\hspace{8mm} \times e^{\textsf{j} \frac{2 \pi}{{N}}\big(-(\operatorname{loc}_p+n) d_p+{N} c_1 d_p^2\big)} .
\end{align}     
Therefore, the norm of $\left\|(\mathbf{\Phi}_k^r (\mathbf{x}^{\prime})-\mathbf{\Phi}_j^r(\mathbf{x}))\right\|_{\mathrm{F}}^2$ in (\ref{Okj_o}) can be calculated as (\ref{F-norm}), shown at the bottom of last page, where $x_{\mathrm{loc}_p+n}^{\prime}$ and $x_{\mathrm{loc}_p+n}$ are the elements of $\mathbf{x}^{\prime}$ and $\mathbf{x}$, respectively. We define $x_{\mathrm{loc}_p+n} = r_{p,n} e^{j\vartheta_{p,n}}$ and $x_{\mathrm{loc}_p+n}^{\prime} =  r_{p,n}^{\prime} e^{j\vartheta_{p,n}^{\prime}}$.

	Utilizing the formula \setcounter{equation}{80}
\begin{equation} % eq.81
	\mid a-b\mid^2 = \mid a\mid^2-2\Re(a b^*)+\mid b\mid^2,
\end{equation}
where $a$ and $b$ are any complex numbers, (\ref{F-norm}) can be expressed as
\begin{align}\label{F-norm68} % eq.82
	& \left\|(\mathbf{\Phi}_k^r (\mathbf{x}^{\prime})\! -\! \mathbf{\Phi}_j^r(\mathbf{x}))\right\|_{\mathrm{F}}^2 
	  \! =\! \sum_{p=1}^P \sum_{n=1}^{N} \left|e^{\textsf{j} \theta_n^{\prime}} x^{\prime}_{\operatorname{loc}_p+n}\! -\! e^{\textsf{j} \theta_n} x_{\operatorname{loc}_p+n} \right|^2 \nonumber \\
	&\hspace*{3mm} =\! \sum_{p=1}^P \sum_{n=1}^{N}\! \! \left(\left|x_{\mathrm{loc}_p+n}\right|^2 + \left|x_{\mathrm{loc}_p+n}^{\prime}\right|^2 \! \right. \nonumber \\
	&\hspace*{20mm}\left. - 2\mathfrak{R}\left(x_{\operatorname{loc}_p+n}^{\prime} \big(x_{\operatorname{loc}_p+n}\big)^* e^{\textsf{j}(\theta_n^{\prime}-\theta_n)}\right)\! \right)\! \nonumber \\
	&\hspace*{3mm}  =\! \sum_{p=1}^P \sum_{n=1}^{N}\! \! \left(r_{p,n}^2 + {r_{p,n}^{\prime \; 2}} \! \right. \nonumber \\
	&\hspace*{20mm}  \left. - 2 r_{p,n}r_{p,n}^{\prime}  \mathfrak{R}\left( e^{\textsf{j}\left(\theta_n^{\prime}-\theta_n + \vartheta_{p,n}^{\prime} - \vartheta_{p,n}\right)}\right)\! \right)\! ,\!
\end{align}
where $\theta_n^{\prime}$ and $\theta_n$ are given by (\ref{theta_define}).

	According to Euler's formula, we further simplify (\ref{F-norm68}) as
\begin{align}\label{FFFF} % eq.83
	&\left\|(\mathbf{\Phi}_k^r (\mathbf{x}^{\prime})\! -\! \mathbf{\Phi}_j^r(\mathbf{x}))\right\|_{\mathrm{F}}^2\! \hspace*{3mm}  = \sum_{p=1}^P \sum_{n=1}^{N}  F_{p,n} ,
\end{align} 				
where $F_{p,n}$ is given by (\ref{eq-Fn}). Substituting (\ref{FFFF}) into (\ref{Okj_o}) leads to $O_{k,j}(\mathcal{P}_c)$ of (\ref{Okj}).

\subsection{Relationship between Speed and Normalized Doppler Shift}\label{calc_doppler}

Let the maximum speed of the mobile station be $v_e$, the carrier frequency be $f_c$, and the speed of light be denoted by $c$. Then the maximum Doppler frequency $\nu_{\max}$ is given by
\begin{equation} % eq.84
	\nu_{\max} \;=\; \frac{v_e \, f_c}{c}.
	\label{eq:max_doppler}
\end{equation}
As established in Section~\ref{S2}, the maximum normalized Doppler shift is obtained through the process of normalization as 
\begin{align} % eq.84
\alpha_{\max}\! & = \! N \, T \, \nu_{\max} = \frac{\nu_{\max}}{f_s} = \frac{v_{e} \, f_c}{c \, f_s}.
\end{align}
where $f_s$ represents the subcarrier spacing.

%\balance
	
	%Biography text here without a photo.
	
	%\begin{IEEEbiography}[{\includegraphics[width=1in,height=1.25in,clip,keepaspectratio]{fig1.png}}]{IEEE Publications Technology Team}
	%In this paragraph you can place your educational, professional background and research, and other interests.\end{IEEEbiography}
\footnotesize
\bibliographystyle{IEEEtran}   %指定参考文献样式文件为IEEEtran.bst
%\bibliography{AFDM-Journal}   %指定所使用的bib文件为IEEEexample

\begin{IEEEbiography}[{\includegraphics[width=1in,height=1.25in,clip,keepaspectratio]{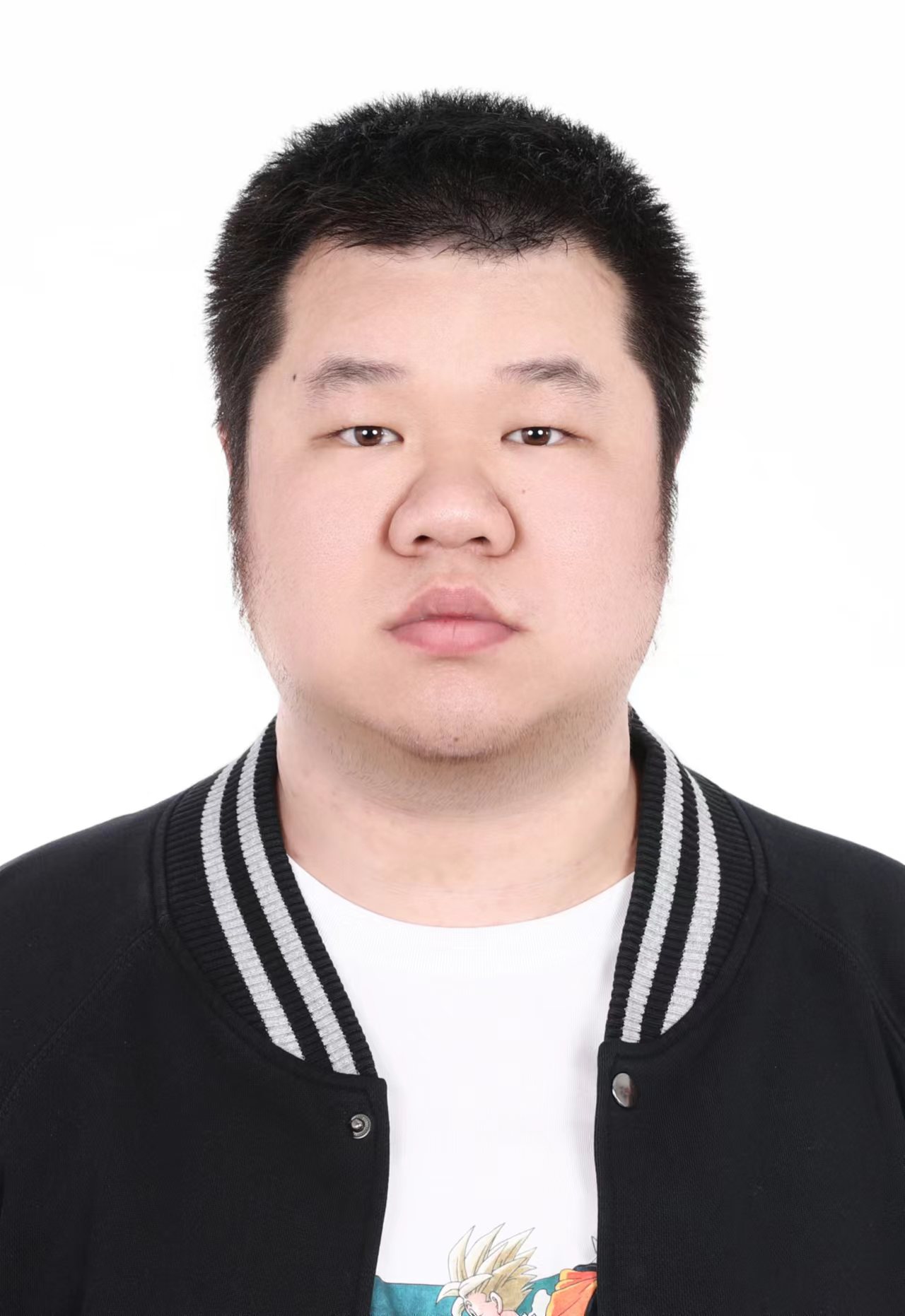}}]{Guangyao Liu}
	
	received the B.S. degree in electronics information engineering from Dalian University of Technology, Dalian, China, in 2019. He is currently working toward the Ph.D. degree with the School of Electronic and Information Engineering, Beihang University. His current research interests include modulation and signal processing for wireless communication and emerging modulation techniques for 6G networks.
	
\end{IEEEbiography}

\begin{IEEEbiography}[{\includegraphics[width=1in,height=1.25in,clip,keepaspectratio]{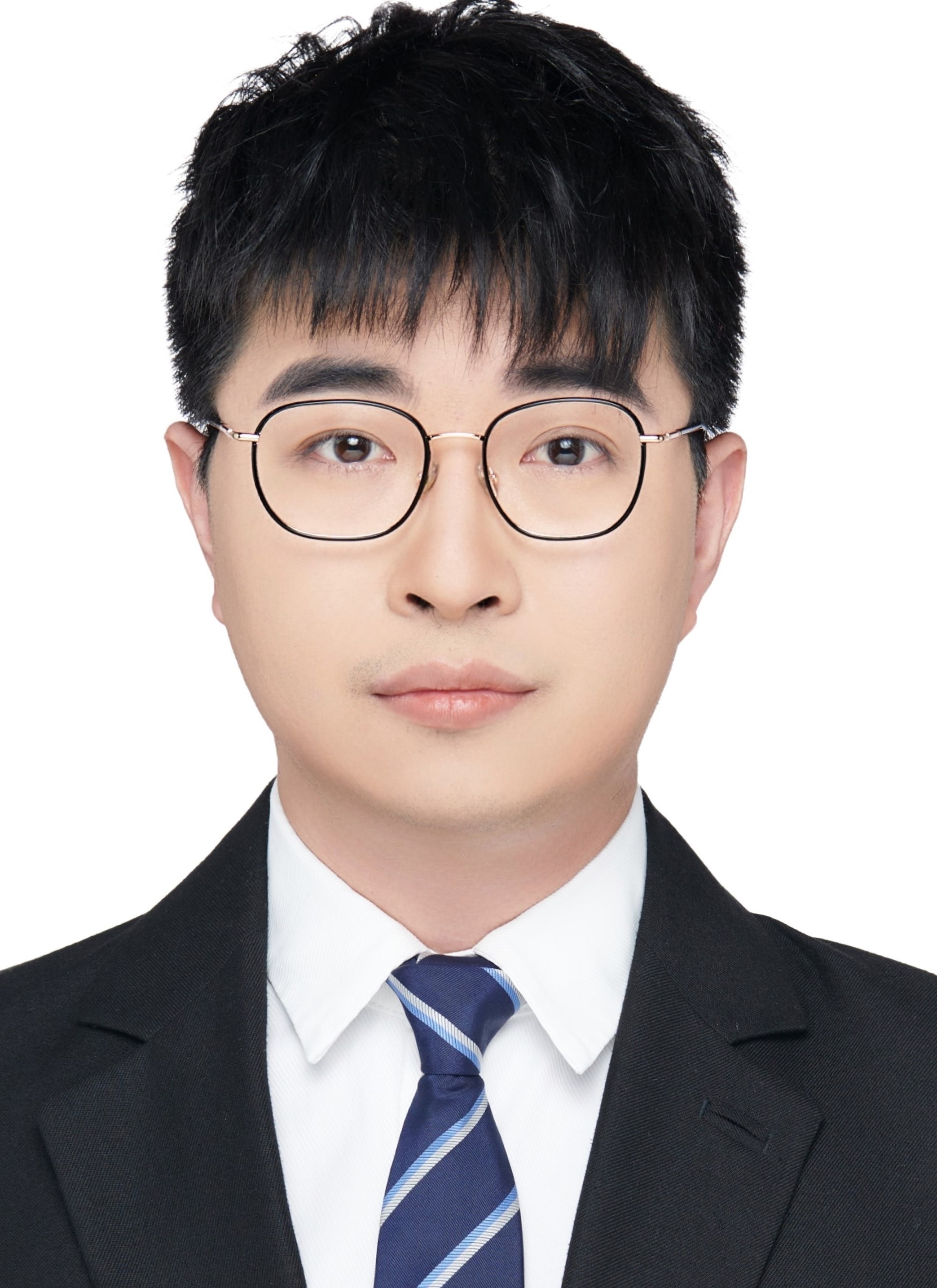}}]{Tianqi Mao}
		(Member, IEEE) received his B.S., M.S. (Hons.), and Ph.D (Hons.) degrees from Tsinghua University in 2015, 2018, and Jan. 2022. He is currently working as an Associate Professor with the Advanced Research Institute of Multidisciplinary Sciences, Beijing Institute of Technology, Beijing, China. He has authored over 50 journal and conference papers, including 3 Highly Cited Papers of ESI (as the first author). His research interests include modulation, waveform design and signal processing for wireless communications, integrated sensing and communication, terahertz communications, and visible light communications. He was a recipient of the Young Elite Scientists Sponsorship Program by China Association for Science and Technology, the Science and Technology Award (Second Prize) of China Institute of Communications and the Outstanding Ph.D. Graduate of Beijing City. He was the winner of IEEE IWCMC 2024 Best Paper Award. He was the organizer of industry panel in IEEE WCNC 2024, and a Tutorial Lecturer in IEEE ICC 2025, IEEE VTC-fall 2024 and IEEE IWCMC 2024. He is currently the Lead Guest Editor of \textsc{IEEE NETWORK} and an Associate Editor of \textsc{IEEE COMMUNICATIONS LETTERS} and \textsc{IEEE OPEN JOURNAL OF VEHICULAR TECHNOLOGY}, and was a Guest Editor of \textsc{IEEE Transactions on Molecular, Biological, and Multi-Scale Communications} and \textsc{IEEE OPEN JOURNAL OF THE COMMUNICATIONS SOCIETY} (two Special Issues). He was also the Exemplary Reviewer of \textsc{IEEE TRANSACTIONS ON COMMUNICATIONS} and \textsc{COMMUNICATIONS LETTERS}
\end{IEEEbiography}
\vspace{-1cm}

\begin{IEEEbiography}[{\includegraphics[width=1in,height=1.25in,clip,keepaspectratio]{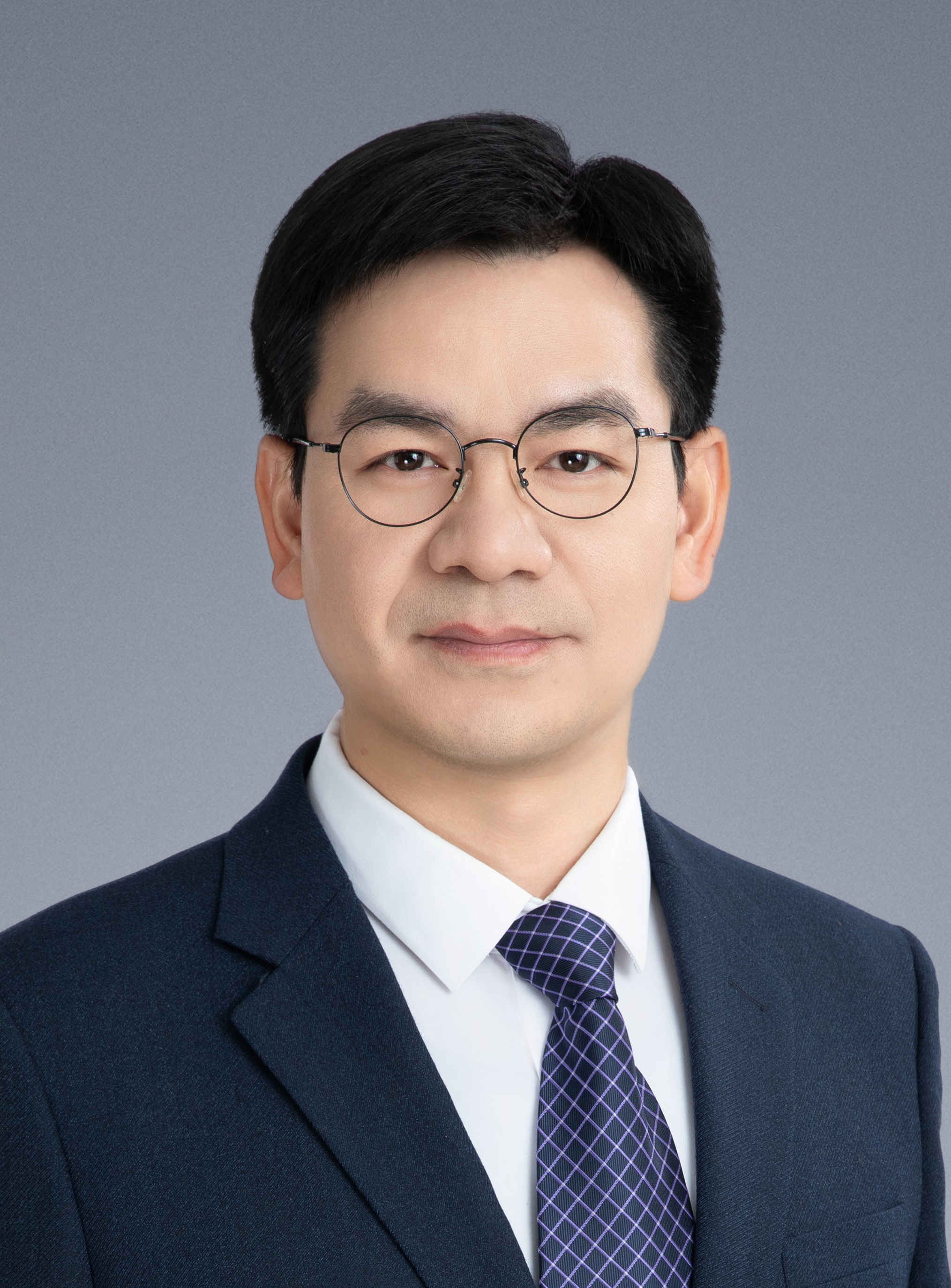}}]{Zhenyu Xiao}
	(Senior Member, IEEE)  received the B.E. degree with the Department of Electronics and Information Engineering, Huazhong University of Science and Technology, Wuhan, China, in 2006, and the Ph.D. degree with the Department of Electronic Engineering, Tsinghua University, Beijing, China, in 2011. From 2011 to 2013, he was a Postdoctorial position with the Department of Electronic Engineering, Tsinghua University. He was with the School of Electronic and Information Engineering, Beihang University, Beijing, as a Lecturer, from 2013 to 2016, and an Associate Professor, from 2016 to 2020, where he is currently a Full Professor. He has visited the University of Delaware from 2012 to 2013 and the Imperial College London from 2015 to 2016. He is currently an Associate Editor for IEEE TRANSACTIONS ON COGNITIVE COMMUNICATIONS AND NETWORKING AND CHINA COMMUNICATIONS. He has also been lead guest Editors of special issues for IEEE JOURNAL ON SELECTED AREAS IN COMMUNICATIONS, China Communications. He was the recipient of the 2017 Best Reviewer Award of IEEE TRANSACTIONS ON WIRELESS COMMUNICATIONS, 2019 Exemplary Reviewer Award of IEEE World Conference on Lung Cancer, and the Fourth China Publishing Government Award, China Youth Science and Technology Award, Second Prize of National Technological Invention, First Prize of Natural Science of China Electronics Society, and the First Prize of Technical Invention of China Society of Aeronautics and Astronautics. He is an active Researcher with broad interests on millimeter wave communications, AAV/satellite communications and networking. He was elected as one of highly cited Chinese researchers since 2020.
\end{IEEEbiography}
\vspace{-1cm}

\begin{IEEEbiography}[{\includegraphics[width=1in,height=1.25in,clip,keepaspectratio]{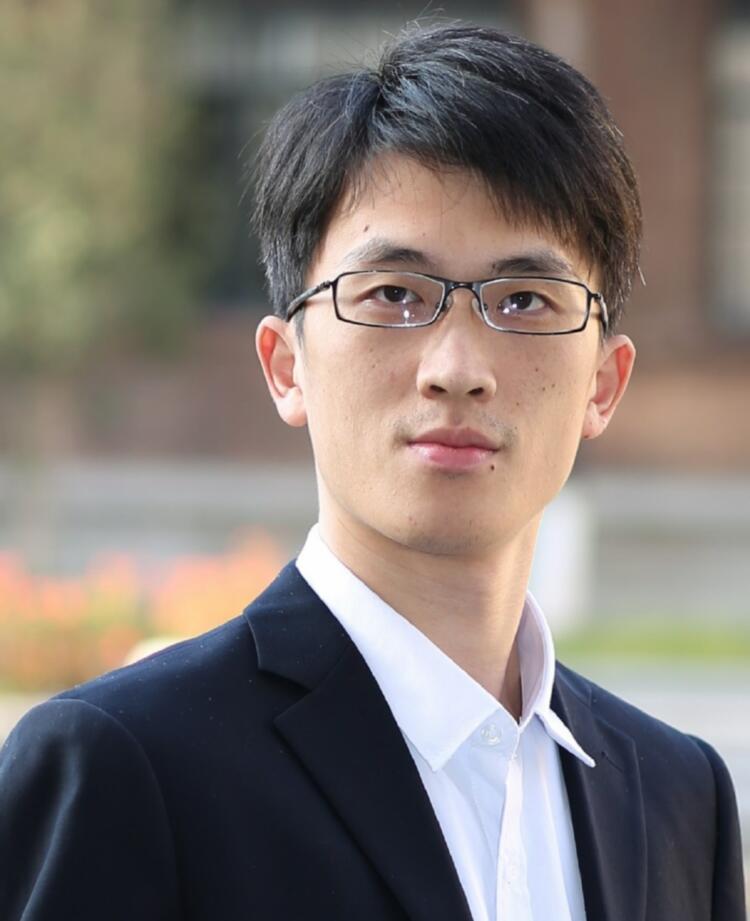}}]{Miaowen Wen}
	
	(Senior Member, IEEE) received the Ph.D. degree in signal and information processing from Peking University, Beijing, China, in 2014. From 2019 to 2021, he was a Hong Kong Scholar with the Department of Electrical and Electronic Engineering, The University of Hong Kong, Hong Kong. He is currently a Professor with the South China University of Technology, Guangzhou, China. He has authored or coauthored two books and more than 200 journal articles. His research interests include a variety of topics in the areas of wireless and molecular communications. He was a recipient of the IEEE Communications Society Asia-Pacific Outstanding Young Researcher Award in 2020. He was the Winner in the Data Bakeoff Competition (Molecular MIMO) from the 2019 IEEE Communication Theory Workshop, Selfoss, Iceland. He served as an Editor for IEEE Transactions on Communications (2019-2024). Currently, he is serving as an Editor/Senior Editor for IEEE Transactions on Wireless Communications, IEEE Transactions on Molecular, Biological, and Multi-scale Communications, and IEEE Communications Letters.
	
\end{IEEEbiography}
\vspace{-.5cm}

\begin{IEEEbiography}
	[{\includegraphics[width=1in,height=1.25in,clip,keepaspectratio]{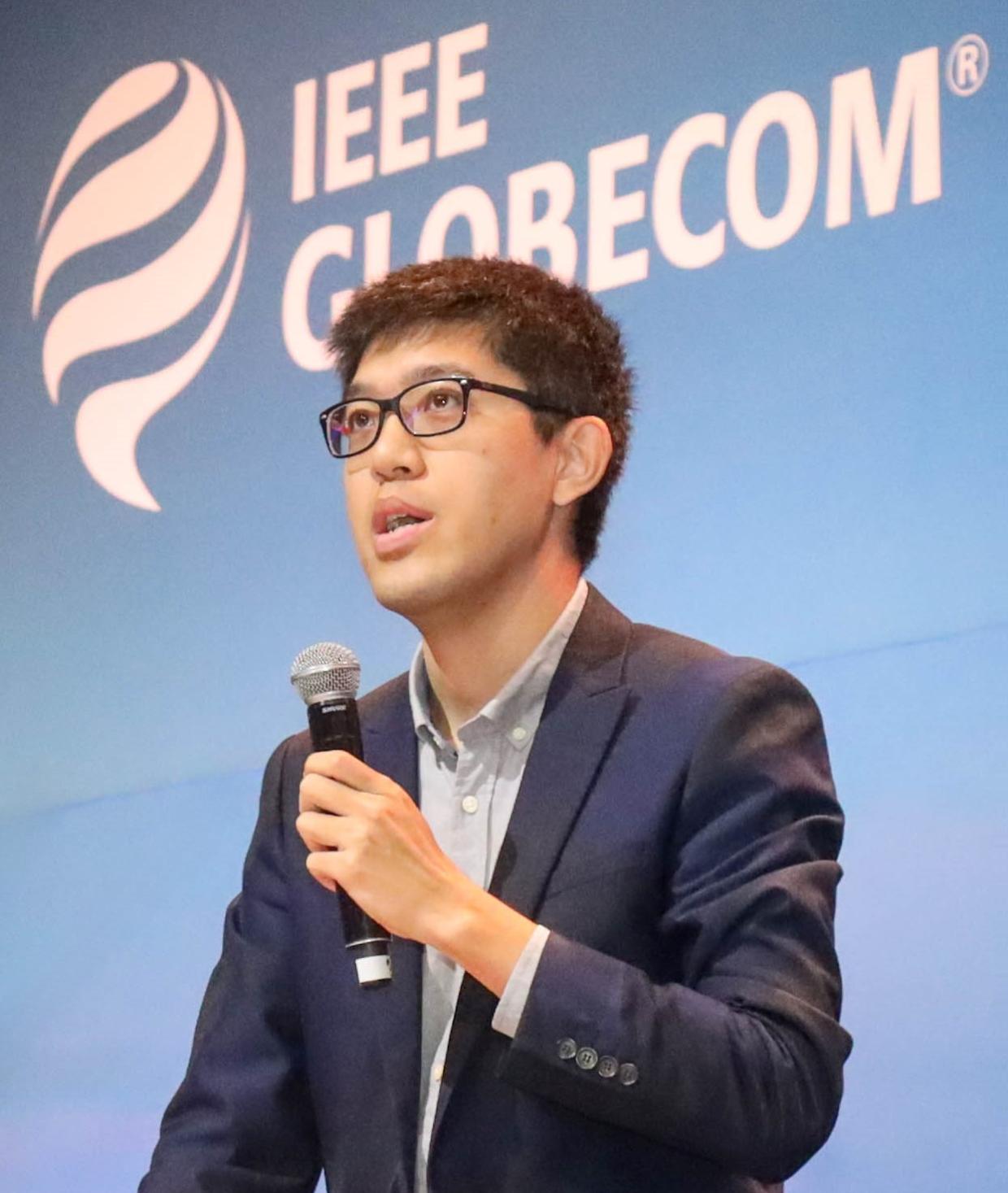}}]{Ruiqi (Richie) Liu} (Senior Member, IEEE) received the B.S. and M.S. degree (with honors) in electronic engineering from the Department of Electronic Engineering, Tsinghua University in 2016 and 2019 respectively. He is now a master researcher in the wireless and computing research institute of ZTE Corporation, responsible for long-term research as well as standardization. His main research interests include reconfigurable intelligent surfaces, integrated sensing and communication and wireless positioning. He currently serves as the Vice Chair of ISG RIS in the ETSI. He is involved in organizing committees of international conferences and is invited to give multiple talks, including the keynote speech at IEEE Globecom 2024. He takes multiple leadership roles in the committees and boards in IEEE ComSoc and VTS, including the voting member of the ComSoc industry communities board. He served as the Deputy Editor-in-Chief of IET Quantum Communication, the Associate Editor for IEEE Communications Letters, the Editor of ITU Journal of Future and Evolving Technologies (ITU J-FET) and Guest Editor or Lead Guest Editor for a series of special issues. His recent awards include the 2022 SPCC-TC Outstanding Service Award and the Beijing Science and Technology Invention Award (Second Prize, 2023).
\end{IEEEbiography}

\begin{IEEEbiography}[{\includegraphics[width=1in,height=1.25in,clip,keepaspectratio]{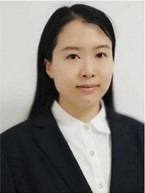}}]{Jingjing Zhao}
	
	(Member, IEEE) received the B.S. degree from the Beijing University of Posts and Telecommunications, Beijing, China, in 2013, and the Ph.D. degree from the Queen Mary University of London, London, U.K., in 2017. From 2017 to 2018, she was a Post-Doctoral Research Fellow with the Department of Informatics, King's College London, London, U.K. From 2018 to 2020, she was a researcher in Amazon, London, U.K. Currently, she is an associate professor with the Beihang University, Beijing, China. Her current research interests include aeronautical broadband communications, RISs aided communications, integrated sensing and communication, and machine learning.  
\end{IEEEbiography}

\begin{IEEEbiography}[{\includegraphics[width=1in,height=1.25in,clip,keepaspectratio]{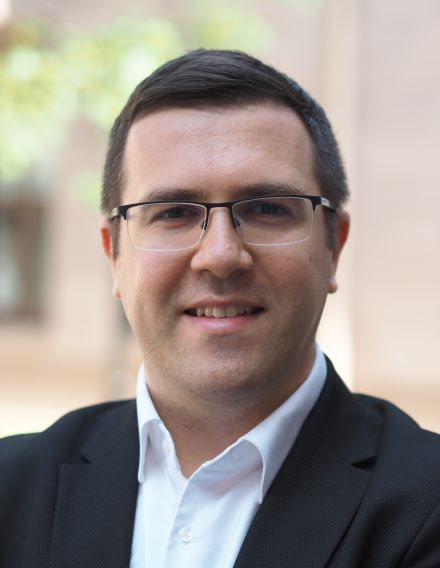}}]{Ertugrul Basar}
	
	(Fellow, IEEE) received the B.S. degree (with High Honors) from Istanbul University, Turkey, in 2007 and the M.Sc. and Ph.D. degrees from Istanbul Technical University, Turkey, in 2009 and 2013, respectively. He is a Professor of Wireless Systems at the Department of Electrical Engineering, Tampere University, Finland. Before joining Tampere University, he held positions at Koç University from 2018 to 2025 and at Istanbul Technical University from 2009 to 2018. He also had visiting positions at Princeton University, Princeton, NJ, USA, as a Visiting Research Collaborator from 2011 to 2022, and at Ruhr University Bochum, Bochum, Germany, as a Mercator Fellow in 2022. Prof. Basar’s primary research interests include 6G and beyond wireless communication systems, multi-antenna systems, index modulation, reconfigurable intelligent surfaces, waveform design, zero-power and thermal noise communications, software-defined radio, physical layer security, quantum key distribution systems, and signal processing/deep learning for communications. He is an inventor of 17 pending/granted patents on future wireless technologies.
	
\end{IEEEbiography}

\begin{IEEEbiography}[{\includegraphics[width=1in,height=1.25in,clip,keepaspectratio]{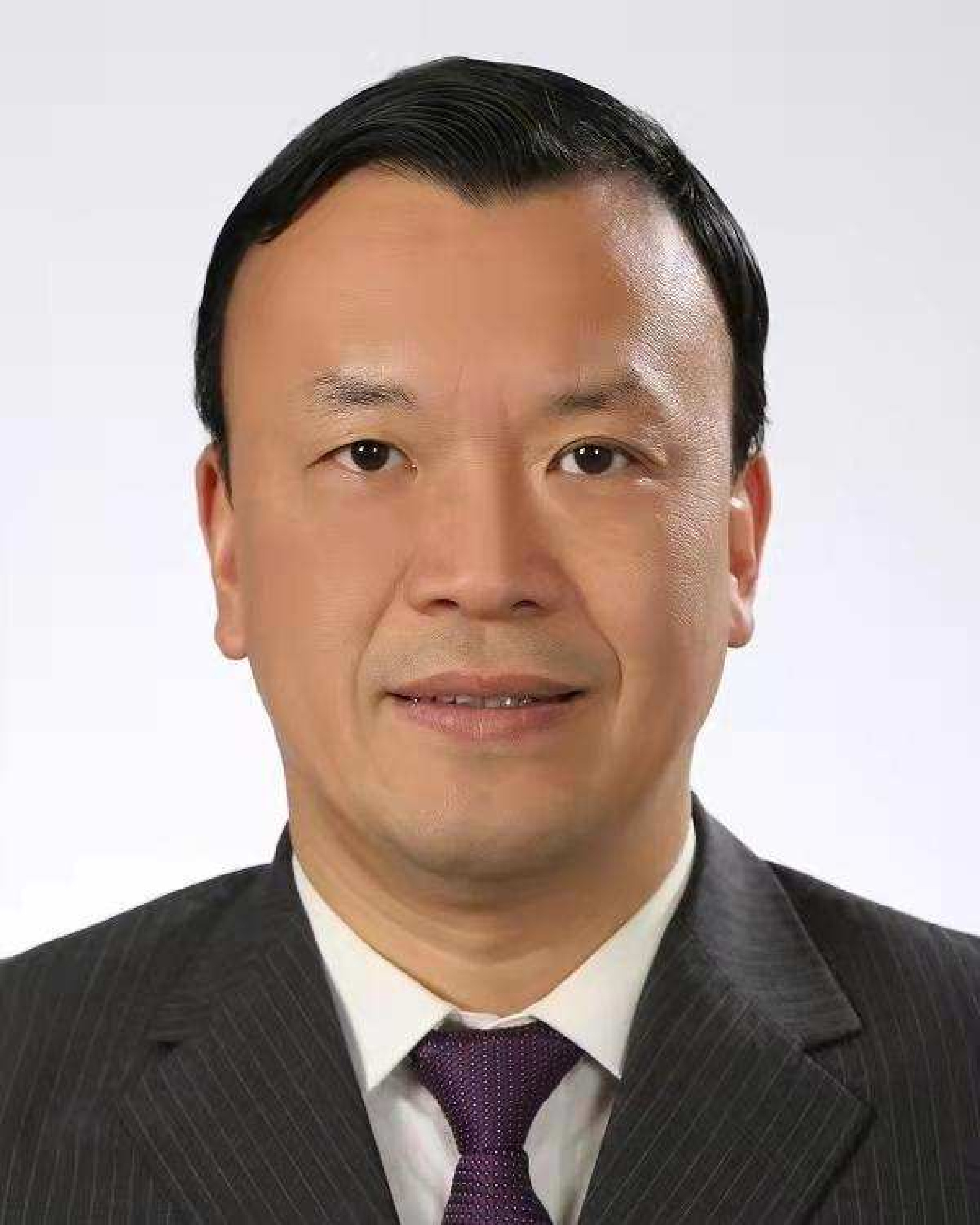}}]{Zhaocheng Wang} (Fellow, IEEE) received the B.S., M.S., and Ph.D. degrees from Tsinghua University in 1991, 1993, and 1996, respectively.   From 1996 to 1997, he was a Post-Doctoral Fellow with Nanyang Technological University, Singapore. From 1997 to 1999, he was a Research Engineer/a Senior Engineer with OKI Techno Centre (Singapore) Pte. Ltd., Singapore. From 1999 to 2009, he was a Senior Engineer/a Principal Engineer with Sony Deutschland GmbH, Germany. Since 2009, he has been a Professor with the Department of Electronic Engineering, Tsinghua University, where he is currently the Director of the Broadband Communication Key Laboratory, Beijing National Research Center for Information Science and Technology (BNRist). He has authored or coauthored two books, which have been selected by IEEE Press Series on Digital and Mobile Communication (Wiley-IEEE Press). He has also authored/coauthored more than 200 peer-reviewed journal articles. He holds 60 U.S./EU granted patents (23 of them as the first inventor). His research interests include wireless communications, millimeter wave communications, and optical wireless communications.  Prof. Wang is a fellow of the Institution of Engineering and Technology. He was a recipient of the ICC2013 Best Paper Award, the OECC2015 Best Student Paper Award, the 2016 IEEE Scott Helt Memorial Award, the 2016 IET Premium Award, the 2016 National Award for Science and Technology Progress (First Prize), the ICC2017 Best Paper Award, the 2018 IEEE ComSoc Asia–Pacific Outstanding Paper Award, and the 2020 IEEE ComSoc Leonard G. Abraham Prize.
\end{IEEEbiography}

\vspace{-3cm}

\begin{IEEEbiography}[{\includegraphics[width=1in,height=1.25in,clip,keepaspectratio]{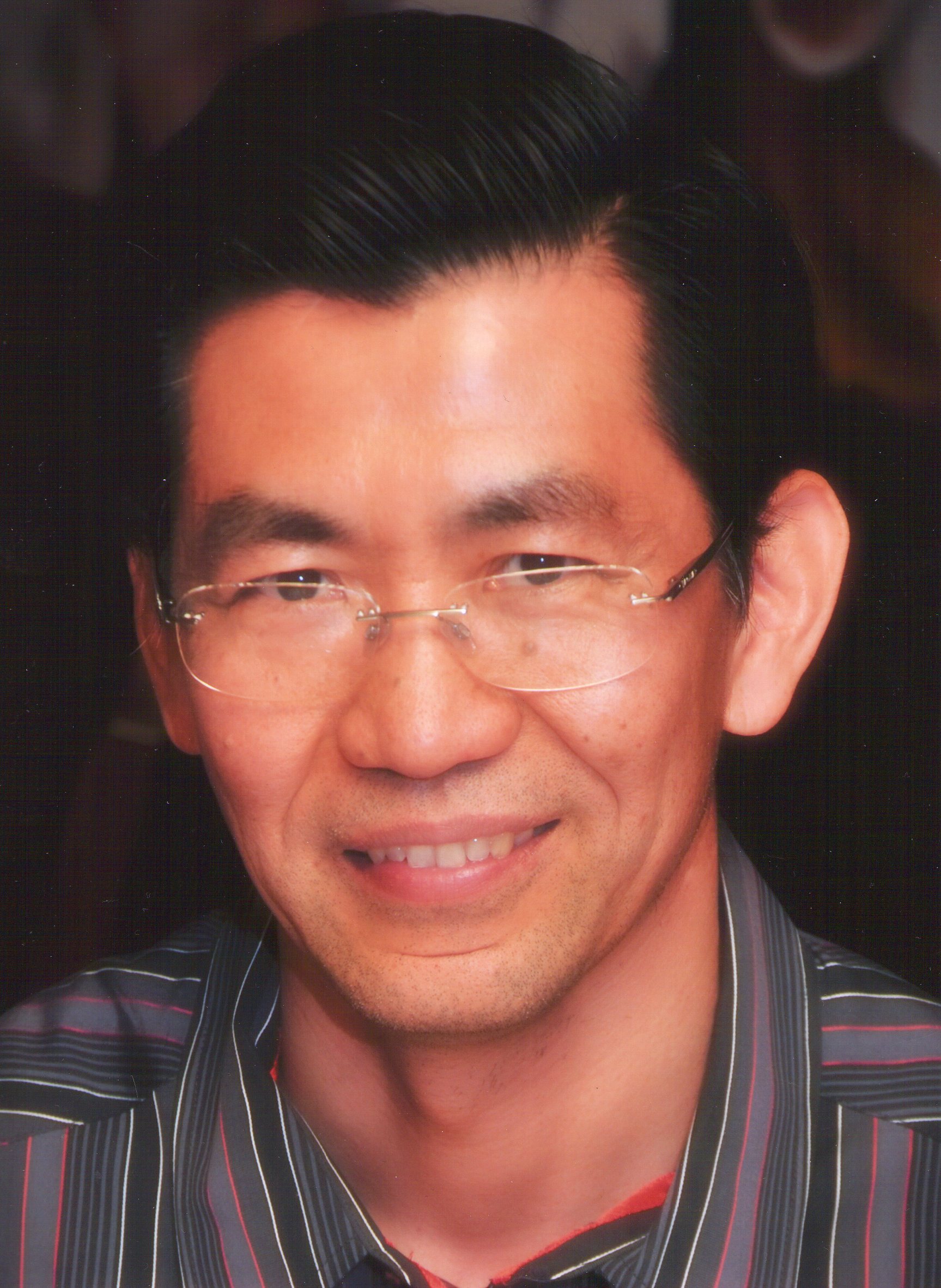}}]{Sheng Chen}
	
	(Life Fellow, IEEE) received his BEng degree from the East China Petroleum Institute, Dongying, China, in January 1982, and his PhD degree from the City University, London, in September 1986, both in control engineering. In August 2005, he was awarded the higher doctoral degree, Doctor of Sciences (DSc), from the University of Southampton, Southampton, UK. From 1986 to 1999, He held research and academic appointments at the Universities of Sheffield, Edinburgh and Portsmouth, all in UK. Since 1999, he has been with the School of Electronics and Computer Science, the University of Southampton, UK, where he holds the post of Professor in Intelligent Systems and Signal Processing. Dr Chen's research interests include adaptive signal processing, wireless communications, modeling and identification of nonlinear systems, neural network and machine learning, evolutionary computation methods and optimization. He has published over 700 research papers. Professor Chen has 21,000+ Web of Science citations with h-index 63 and 40,000+ Google Scholar citations with h-index 86. Dr. Chen is a Fellow of the United Kingdom Royal Academy of Engineering, a Fellow of Asia-Pacific Artificial Intelligence Association and a Fellow of IET. He is one of the 200 original ISI highly cited researchers in engineering (March 2004).
	
\end{IEEEbiography}

\end{document}